\documentclass[journal,twoside,web]{ieeecolor}
\usepackage{generic}
\usepackage{cite}
\usepackage{amsmath,amssymb,amsfonts}
\usepackage{algorithmic}
\usepackage{graphicx}
\usepackage{textcomp}
\usepackage{amsmath}
\usepackage{amssymb}
\usepackage{amsfonts}
\usepackage{graphicx}
\usepackage{enumerate}
\usepackage{mathrsfs}
\usepackage{float}
\usepackage{cite}
\usepackage{tikz}
\usetikzlibrary{matrix,arrows,decorations.pathmorphing}
\usepackage{verbatim}
\usepackage{mathtools}
\usepackage{caption}
\usepackage{subcaption}
\usepackage{soul}
\setstcolor{magenta}
\usetikzlibrary{arrows}
\usepackage{tabularx}
\newtheorem{assumption}{Assumption}

\newtheorem{definition}{Definition}
\newtheorem{proposition}{Proposition}
\newtheorem{theorem}{Theorem}
\newtheorem{lemma}{Lemma}
\newtheorem{remark}{Remark}

\newcommand{\csx}[1]{\textcolor{blue}{#1}}
\newcommand{\zcc}[1]{\textcolor{blue}{#1}}

\def\BibTeX{{\rm B\kern-.05em{\sc i\kern-.025em b}\kern-.08em
    T\kern-.1667em\lower.7ex\hbox{E}\kern-.125emX}}
\begin{document}
\title{\LARGE \bf
Stability and Robustness of Tensor-Coupled  Flow-Conservation Dynamical Systems on Hypergraphs }

\author{Chencheng Zhang$^{1}$, Hao Yang$^{1}$, Bin Jiang$^{1}$, \IEEEmembership{Fellow, IEEE,} Shaoxuan Cui$^{2}$ 
 \thanks{This work was supported in part by the National Natural Science Foundation of China under Grants 62073165, 62233009, and 62403241, in part by Postdoctoral Fellowship Program of CPSF under Grants GZB20240975, 2025T181136, 2025M774280, in part by Jiangsu Provincial Outstanding Postdoctoral Program, and in part by the National Key Laboratory Foundation of Helicopter Aeromechanics. }
\thanks{$^{1}$ C. Zhang, H. Yang and B. Jiang are with the College of Automation Engineering, Nanjing University of Aeronautics and Astronautics, Nanjing, 211106, China (e-mail: zhangchencheng@nuaa.edu.cn; haoyang@nuaa.edu.cn; binjiang@nuaa.edu.cn). }
\thanks{$^{2}$ S. Cui is with the Bernoulli Institute for Mathematics, Computer Science and Artificial Intelligence, University of Groningen, Groningen, 9747 AG Netherlands (E-mail: s.cui@rug.nl)}}

\maketitle

\begin{abstract}
This paper develops an entropy-based stability and robustness framework for nonlinear hypergraph dynamics with conservation and flow balance. 
We consider generator-form systems on the  simplex whose state-dependent transition rates capture higher-order (tensor) interactions among nodes. 
Under a \csx{tensor generalized detailed-balance (TGDB) condition}, we show that the system admits a unique equilibrium and an entropy Lyapunov function ensuring global asymptotic stability. 
The Jacobian restricted to the tangent subspace of the simplex is Hurwitz, and its spectral gap determines the exponential convergence rate. 
Building on this structure, we derive first-order sensitivity bounds of the equilibrium under perturbations of the coupling tensor and establish a local input-to-state stability (ISS) estimate with respect to external inputs. 
The results reveal a quantitative link between the spectral gap and the system's robustness margin: larger spectral gaps imply smaller equilibrium shifts and faster recovery under structural or parametric perturbations. 
Numerical experiments on tensor-coupled flow models confirm the theoretical predictions and illustrate how the proposed entropy-dissipating framework unifies stability and robustness analysis for conservative higher-order network systems.
\end{abstract}


\section{Introduction}
Many real-world networked systems, ranging from social opinion formation~\cite{neuhauser2022consensus} and ecological population evolution~\cite{cui2025analysis} to multi-agent coordination~\cite{neuhauser2021consensus, zhang2023flocking} and epidemic spreading~\cite{cui2025sis, liang2025discrete}, involve interactions that are not limited to pairwise links but occur through higher-order group relations~\cite{bick2023higher, Majhi2022DynamicsHigherOrder, Faccin2022DirectedHypergraphs, Zhang2023}.
Such collective interactions can be naturally represented by \emph{hypergraphs}, whose edges (hyperedges) connect arbitrary subsets of nodes.
While traditional graph-based models capture diffusion or influence between pairs of agents, hypergraph frameworks can encode group-level couplings, cooperative or antagonistic multi-node interactions, and higher-order diffusion patterns that cannot be reduced to pairwise forms.
This shift from edges to hyperedges has prompted a new generation of higher-order dynamical models, often expressed through multilinear or tensor coupling structures.

Conservation of a networked system refers to the invariance of global state (e.g., the sum or mean of node states is preserved throughout the trajectories of the system). 
In classical graph-based dynamics such as multi-agent systems, this property naturally arises from the Laplacian's row-sum-zero structure, ensuring that the overall or average position error between agents remains constant over time~\cite{zeng2016convergence, Mesbahi2010}. 
\zcc{Traditional network dynamics on ordinary graphs already form a mature research area. Pairwise couplings underlie consensus and cooperative control in multi-agent networks~\cite{zeng2016convergence, Mesbahi2010}, epidemic and diffusion processes on networks~\cite{cui2022discrete}, and Markovian flow models with sensitivity and entropy structures~\cite{Mitrophanov2005SensitivityErgodicMC, Maas2011GradientEntropy}. Robustness and resilience of such graph-based systems have likewise been studied from disturbance attenuation, ISS, and attack-resilience viewpoints~\cite{Dashkovskiy2007NetworkISS, Pasqualetti2013AttackDetection}. These works provide the classical baseline for our study. The present paper asks how analogous stability and robustness questions should be formulated when the couplings are generated not by pairwise edges but by higher-order tensors on a hypergraph.}
Rather than graph-based dynamics, recent advances have established rigorous formulations for \emph{hypergraph-coupled dynamical systems}, where the evolution of nodal states is governed by tensor-based interaction terms that capture higher-order interactions.
In~\cite{cui2024metzler, cui2024discrete}, they analyzed a broad class of such higher-order systems, and these results were further utilized to study hypergraph Laplacian dynamics~\cite{cui2025higher}, higher-order Lotka--Volterra~\cite{cui2025analysis}, and SIS processes on hypergraphs~\cite{cui2025sis, liang2025discrete}.
All these studies revealed how nonpairwise coupling and higher-order interactions can fundamentally reshape equilibrium patterns, convergence rates, and collective behaviors.
In~\cite{chen2021controllability, pickard2023observability, zhang2024global, dong2024controllability}, both controllability and observability are further investigated.
These works demonstrated that tensor algebra is a powerful tool for analyzing hypergraph systems.
However, the above modeling frameworks of tensor-coupled hypergraph systems and other studies of higher-order networks~\cite{battiston2020networks, Neuhauser2020, Zhang2023} rarely consider whether such conservation actually holds.
When interactions extend beyond pairwise couplings, higher-order terms can also induce conservation of quantities, meaning that the total sum of states can also be invariant.
As a result, the influence of conservation properties on the stability, convergence, and collective behavior remains largely unexplored in the context of higher-order or tensor-based networked systems.

Building upon these insights, the present work turns to the complementary class of \emph{conservative} hypergraph systems, in which the total mass (or probability, flow) is preserved and the dynamics satisfy a flow-balance condition analogous to reversible Markov systems~\cite{Maas2011GradientEntropy, Erbar2012RicciMarkov}.
In this regime, conservation introduces additional geometric constraints on the state space and enables the emergence of entropy-like Lyapunov functions, leading to an entropy-dissipating characterization of system stability.
This structural perspective unifies higher-order hypergraph dynamics with a class of nonlinear, reversible flow systems, forming the basis for the robustness analysis developed in our work.

However, while the stability of conservative hypergraph systems has received increasing attention, much less is known about their \emph{robustness}.
In contrast to classical networked systems, where linear or pairwise diffusive couplings allow spectral tools to quantify resilience~\cite{Mitrophanov2005SensitivityErgodicMC}, higher-order systems are governed by nonlinear, state-dependent transition rates that make standard spectral analysis inapplicable.
Consequently, questions such as how equilibrium states shift under perturbations of the hypergraph tensor, how convergence rates relate to structural properties, and how entropy dissipation reacts to external inputs remain largely open.
Furthermore, realistic networked systems inevitably encounter perturbations, communication faults, and external disturbances~\cite{LeBlanc2013ResilientConsensus, Pasqualetti2013AttackDetection}, highlighting the need for robust analytical frameworks.
Classical robustness concepts such as Input-to-State Stability (ISS)~\cite{sontag2008input, Jiang1994SmallGainISS, Zhang2022FractionalFTC} and small-gain conditions for interconnected networks~\cite{Dashkovskiy2007NetworkISS, Mironchenko2017InfiniteDimISS, Yang2020Exponential} provide useful inspiration, yet they have not been extended to the nonlinear, tensor-coupled conservative setting.

This paper closes this gap by establishing an entropy-based robustness theory for nonlinear hypergraph dynamics with conservation and flow balance.
Our main contributions are:

\textit{1) Modeling and considering stability of generator-form systems on the simplex:}
The dynamics are formulated as generator-form systems on the simplex, where state-dependent transition rates capture higher-order (tensor) interactions among nodes.
Under a generalized detailed-balance (GDB) condition, the system admits a unique equilibrium and exhibits global asymptotic stability.
This modeling strategy enables rigorous treatment of complex high-order interactions while preserving analytical tractability.

\textit{2) Developing an entropy Lyapunov framework for nonlinear hypergraph dynamics:}
When the dynamics incorporate both state conservation and flow balance, the first general analysis of stability for tensor-coupled high-order network systems under such constraints is provided.
This approach overcomes the limitations of traditional energy- or linear-Lyapunov-based methods, which are generally inapplicable to nonlinear high-order interactions.

\textit{3) Analyzing robustness via spectral gap and perturbation theory:}
Based on the equilibrium structure, a regular perturbation analysis derives first-order sensitivity bounds with respect to structural perturbations, and the system is shown to satisfy a local ISS estimate under external excitations.
A direct quantitative link is established between the \emph{spectral gap} of the Jacobian and the system's robustness margin.


Overall, the paper contributes to the emerging field of higher-order network dynamics by unifying stability and robustness analysis within a common entropy-dissipating framework. 
It generalizes classical detailed-balance systems to the nonlinear, tensor-coupled hypergraph setting, offering new theoretical tools for quantifying sensitivity and resilience in high-dimensional conservative flows.

\zcc{The rest of the paper is organized as follows. Section~II introduces the notation, preliminaries, and hypergraph background. Section~III formulates the tensor-coupled flow-conservation dynamics and establishes the entropy-equilibrium structure together with the stability results. Section~IV develops the robustness analysis, including equilibrium sensitivity and local ISS bounds under perturbations. Section~V presents numerical experiments, and Section~VI concludes the paper.}

\section{Preliminaries and Notations}

\subsection{Basic notation}
For an integer $n\ge 1$, we write $[n]:=\{1,2,\dots,n\}$. 
Vectors are column vectors unless stated otherwise, and $\mathbf{1}$ denotes the all-ones vector in $\mathbb{R}^n$. 
The standard simplex is
\[
\Delta_n := \{x\in\mathbb{R}^n_{>0} \mid \mathbf{1}^\top x = 1\},
\]
whose tangent subspace is $\mathbf{1}^\perp := \{z\in\mathbb{R}^n \mid \mathbf{1}^\top z = 0\}$.
For a differentiable vector field $f:\mathbb{R}^n\to\mathbb{R}^n$, the Jacobian at $x$ is denoted 
$J(x) := \partial f / \partial x |_{x}$.
We use $\|\cdot\|$ for the Euclidean norm and $\|\cdot\|_F$ for the Frobenius norm of a matrix or tensor.
For a square matrix $M$, $\sigma(M)$ denotes its spectrum, and $\mathrm{Re}(\lambda)$ the real part of an eigenvalue~$\lambda$.
\zcc{A continuous function $\alpha:\mathbb{R}_{\ge 0}\to\mathbb{R}_{\ge 0}$ is said to be of class $\mathcal{K}$ if it is strictly increasing and satisfies $\alpha(0)=0$. It is of class $\mathcal{K}_\infty$ if, in addition, $\alpha(s)\to\infty$ as $s\to\infty$. A continuous function $\beta:\mathbb{R}_{\ge 0}\times\mathbb{R}_{\ge 0}\to\mathbb{R}_{\ge 0}$ is of class $\mathcal{KL}$ if, for each fixed $t\ge 0$, the map $\beta(\cdot,t)$ is of class $\mathcal{K}$ and, for each fixed $s\ge 0$, the map $\beta(s,\cdot)$ is nonincreasing and converges to zero as $t\to\infty$.}

\begin{definition}[Local input-to-state stability (LISS) {\cite{mironchenko2016local}}]
\label{def:liss}
For the nonlinear system
\[
\dot{x}(t) = f(x(t),u(t)),
\]
where $x \in \mathbb{R}^n$, $u \in \mathbb{R}^m$ are locally bounded, and
$f:\mathbb{R}^n \times \mathbb{R}^m \rightarrow \mathbb{R}^n$
is continuous and locally Lipschitz in $x$,
the system is said to be locally input-to-state stable (LISS) with respect to $u(t)$ if there exist a class $\mathcal{KL}$ function $\beta$, a class $\mathcal{K}$ function $\gamma$, and a constant $r>0$ such that
\begin{equation}
\setlength\abovedisplayskip{7pt}
\setlength\belowdisplayskip{7pt}
|x(t)| \le \beta(|x(0)|,t) + \gamma(\|u\|_{[0,t)}),
\label{eq:liss-def}
\end{equation}
for all $t \ge 0$, all initial states $x(0)$ satisfying $|x(0)| \le r$, and all inputs $u(t)$ satisfying $\|u\|_{[0,t)} \le r$.
\end{definition}

\begin{lemma}[Local Lyapunov criterion for LISS]
\label{lem:liss}
Consider the nonlinear system $\dot{x}=f(x,u)$ with $x\in\mathbb{R}^n$, $u\in\mathbb{R}^m$, where $f$ is continuous and locally Lipschitz in $x$.
Suppose there exist a continuously differentiable function $V:\mathbb{R}^n\to\mathbb{R}_{\ge 0}$, a constant $r>0$, class $\mathcal{K}_\infty$ functions $\alpha_1,\alpha_2$, and class $\mathcal{K}$ functions $\alpha_3,\sigma$ such that for all $(x,u)\in\mathbb{R}^n\times\mathbb{R}^m$ satisfying $|x|\le r$ and $|u|\le r$,
\begin{align}
&\alpha_1(|x|)\le V(x)\le \alpha_2(|x|), \label{eq:liss-lyap-bound}\\
&\dot V(x,u)\le -\,\alpha_3(|x|)+\sigma(|u|). \label{eq:liss-lyap-diss}
\end{align}
Then the state $x(t)$ of the system is locally input-to-state stable with respect to $u(t)$.
\end{lemma}



\subsection{Directed, weighted, and nonuniform hypergraphs}
A \emph{hypergraph} is a pair $\mathcal{H}=([n],\mathcal{E})$ where $\mathcal{E}$ is a finite collection of nonempty subsets $e\subseteq[n]$ called \emph{hyperedges}.  
Each hyperedge $e$ may carry a weight $w_e>0$ and an orientation, written $e=(T_e,H_e)$, where $T_e\subseteq[n]$ and $H_e\subseteq[n]$ denote respectively the \emph{tail} (source) and \emph{head} (target) node sets.  
The hypergraph is \emph{directed} if some $e$ have $T_e\neq H_e$, and \emph{nonuniform} if hyperedges have varying cardinalities $|T_e|,|H_e|$. In this paper, we only focus on directed hypergraphs with one head \csx{for the same reason as stated in \cite{cui2024metzler,cui2025sis}. Our main interest is in nodal dynamics on hypergraphs. From this viewpoint, each hyperedge represents the joint influence of a group of nodes on the state evolution of a single target node, so the one-head formulation is sufficient for describing the dynamics at the node level.} We call $\mathcal{H}$ \emph{weighted} if it is equipped with a weight function $w:\mathcal{E}\to\mathbb{R}_{>0}$. 

\csx{A general non-uniform hypergraph can be naturally decomposed into a multilayer structure \cite{cui2024metzler}, where each layer collects hyperedges of the same cardinality and therefore forms a uniform hypergraph. In this viewpoint, the original hypergraph is represented as the superposition of several uniform sub-hypergraph components, each associated with one interaction order. Such a decomposition is useful because it separates higher-order interactions by degree, allowing the overall system to be analyzed through its order-specific subnetworks while still preserving the full non-uniform structure at the aggregate level.}
\csx{Another common representation is \emph{graph projection} \cite{cui2024metzler}, in which a hypergraph is mapped into an ordinary graph so that classical graph-based tools can be applied. However, this transformation typically reduces higher-order interactions to pairwise relations and may therefore lead to a loss of higher-order structural information. The both representation approaches are illustrated via Fig. \ref{fig:ill}.}

\begin{figure}
        \centering
        \includegraphics[height=5cm]{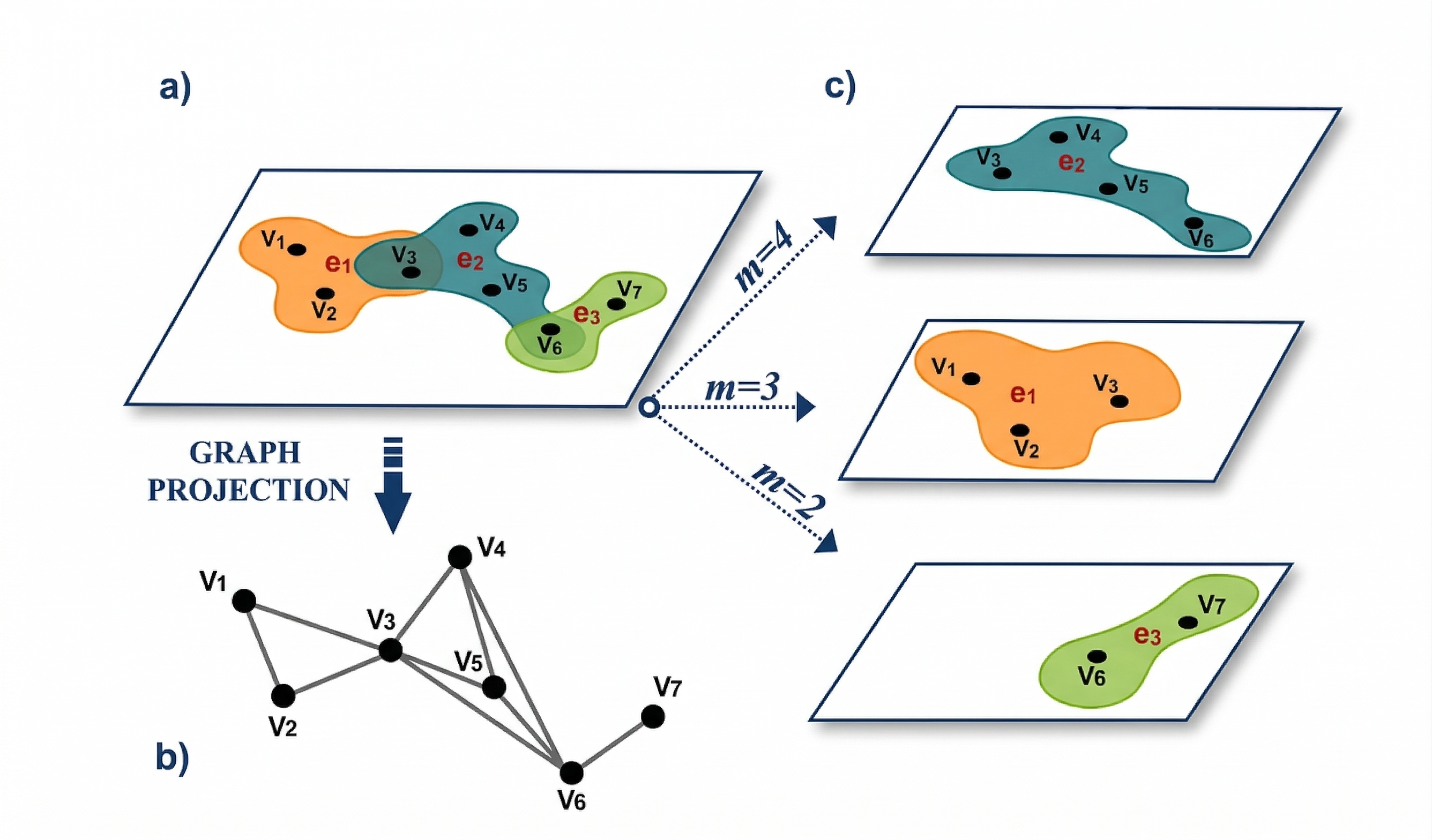}
        \caption{\csx{Illustration figure of a non-uniform undirected hypergraph. The original hypergraph is illustrated in a). The decomposition into different layers with each layer a uniform hypergraph is illustrated in b). The Hypergraph can be projected as a graph as c).}}
        \label{fig:ill}
\end{figure}

\subsection{Adjacency tensors and interaction encoding}
To algebraically represent higher-order interactions, we associate with $\mathcal{A}$ a collection of \emph{adjacency tensors} 
\[
\mathcal{A}^{(r)}=\{a^{(r)}_{i_1 i_2 \dots i_{r+1}}\}\in\mathbb{R}^{n^{r+1}}_{\ge 0}, 
\quad r=1,2,\dots,R,
\]
where each entry $a^{(r)}_{i_1 i_2\dots i_{r+1}}$ quantifies the contribution of the $r$-th order hyperedge connecting tail nodes $(i_2,\dots,i_{r+1})$ towards the head $i_1$.
The order $r$ indicates that the interaction involves $r+1$ nodes simultaneously.  
For instance, $r=1$ corresponds to pairwise edges (a standard adjacency matrix), and $r=2$ describes three-body interactions.  
If all nonzero entries satisfy $a^{(r)}_{i_1 i_2\dots i_{r+1}}=a^{(r)}_{i_{\pi(1)} i_{\pi(2)} \dots i_{\pi(r+1)}}$ for any permutation $\pi$, the tensor is called \emph{symmetric} and the corresponding hypergraph is undirected.



\csx{Given a collection $\mathcal{A}=\{\mathcal{A}^{(r)}\}_{r=1}^R$, each tensor
$\mathcal{A}^{(r)}=\{a^r_{i k I}\}$
describes one interaction layer of order $r$, where $i$ is the head node, $k$ is a distinguished tail node, and
$I=(i_1,\dots,i_{r-1})\in[n]^{r-1}$ denotes the remaining tail nodes participating in the same higher-order interaction.
Hence, $a^r_{i k I}$ quantifies how the joint state of the tail nodes $(k,i_1,\dots,i_{r-1})$ contributes to the influence toward node $i$.
For the pairwise case, there are no additional tail nodes beyond $k$, so $I$ is understood as the empty set and the corresponding coefficients reduce to a matrix form.
More generally, allowing multiple orders $r=1,\dots,R$ corresponds to a nonuniform hypergraph, whereas a single fixed order corresponds to a uniform hypergraph.}

\section{Tensor-Coupled Dynamics on Hypergraphs}

\subsection{Model formulation and generator representation}

Many distributed systems can be described as \emph{state-dependent
flow-conservation dynamics}, where each node~$i$ carries a fraction~$x_i$ of
a conserved quantity such as probability mass, population, or resource share.
Group interactions between nodes redistribute this quantity without changing
the total amount.  To capture such higher-order interactions, we model the
system on a (possibly directed, weighted, and nonuniform) hypergraph
$\mathcal{H}=([n],\mathcal{E})$ with node set $[n]=\{1,\dots,n\}$.

Denote by $x\in\Delta_n$ the state vector on the simplex
\begin{equation}\label{eq:simplex}
\Delta_n := \{x\in\mathbb{R}^n_{>0} \mid \mathbf{1}^\top x = 1\},
\end{equation}
where $\mathbf{1}^\top x=1$ represents the global conservation law. 

\begin{remark}
In general, the conserved quantity need not equal one, as it is determined by the
initial condition $\mathbf{1}^\top x(0)$ and remains constant along trajectories.
 The particular choice of the total sum (e.g., $1$) is not essential; it only sets a reference scale. Taking the multi-agent system as an example. When the sum is normalized to one, each state $x_i$ represents the relative proportion (or share) of the total quantity corresponding to agent 
$x_i$.
\end{remark}

\paragraph{Tensor formulation}
Higher-order interactions are modeled through polynomial (tensor) couplings.
For each ordered pair $(i,k)$ with $i\neq k$, let
$\mathcal{I}_{r}\subseteq [n]^{r-1}$ ($r=2,\dots,R$) denote index sets and
$a_{i k I}\ge 0$ nonnegative coefficients, where $I=(i_1,\dots,i_{r-1})\in\mathcal{I}_r$.
Define the polynomial rate kernel
\begin{equation}\label{eq:Qikx}
Q_{k\to i}(x)
    := \sum_{r=1}^{R}\;\sum_{I\in\mathcal{I}_r}
    a^r_{i k I}\,\prod_{j\in I} x_{j},
\end{equation}
where $r$ determines the hyperedge order and
$\mathcal{A}=\{a^r_{i k I}\}$ encodes the directed, weighted, and higher-order
structure of~$\mathcal{H}$, the superscript $r$ shows the order of the tensor.

\paragraph{Flow-conservation dynamics}
The tensor-coupled dynamics of node states are defined by the balance of
incoming and outgoing flows,
\begin{equation}\label{eq:tensor-dynamics}
\dot x_i = f_i(x;\mathcal{A})=\sum_{k\neq i} \Phi_{k\to i}(x;\mathcal{A}),
\qquad i\in[n],
\end{equation}
where $\Phi_{k\to i}$ is the net flow from node~$k$ to~$i$.
A canonical and physically interpretable choice is
\begin{equation}\label{eq:Phi-def}
\Phi_{k\to i}(x;\mathcal{A})
    = Q_{k\to i}(x)\,x_k - Q_{i\to k}(x)\,x_i.
\end{equation}
Here $Q_{k\to i}(x)$ represents a state-dependent transmission rate determined
by group interactions, while $x_k$ is the available mass at node~$k$.
The first term $Q_{k\to i}(x)x_k$ denotes the inflow to~$i$ from~$k$, and the
second term $Q_{i\to k}(x)x_i$ the outflow from~$i$ to~$k$.

\paragraph{Generator representation}
Define $Q_{k\to i}(x):=Q_{k\to i}(x)\ge 0$. Then~\eqref{eq:tensor-dynamics}
can be rewritten as
\begin{equation}\label{eq:generator}
\dot x_i
    = \sum_{k\neq i}\big(Q_{k\to i}(x)\,x_k - Q_{i\to k}(x)\,x_i\big),
    \qquad i\in[n].
\end{equation}
Equation~\eqref{eq:generator} expresses a \emph{state-dependent flow}:
each node updates its state according to the net difference between inflow and
outflow determined by higher-order interactions.

\csx{The proposed generator-form dynamics can be viewed as a hypergraph analogue of Laplacian flow. In the classical graph setting, Laplacian dynamics are generated by pairwise edge couplings and describe conservative redistribution of mass through edge-based flows. Here, the redistribution mechanism is generalized to higher-order interactions. In this sense, the tensors $\{A^{(r)}\}_{r=1}^R$ may be interpreted as higher-order Laplacian-like objects generating the hypergraph flow layer by layer, while the full dynamics arise from their superposition across different interaction orders. When only the first-order layer $r=1$ is present, the model reduces to the classical graph-based Laplacian flow as a special case.}

\paragraph{Mass conservation and simplex invariance}
The generator form~\eqref{eq:generator} defines a nonlinear conservative
system whose trajectories are constrained to the probability simplex.
The following result formalizes the positivity and invariance properties.

\begin{proposition}[Positivity and simplex invariance]\label{prop:positivity}
Consider the dynamics~\eqref{eq:generator} on
\[
\Delta_n := \{x\in\mathbb{R}^n_{\ge 0}\mid \mathbf{1}^\top x = 1\},
\]
with continuous nonnegative rate functions $Q_{k\to i}(x)\ge0$.
Then:
\begin{enumerate}[(i)]
\item (\emph{Mass conservation}) $\mathbf{1}^\top \dot x = 0$, so that
      $\mathbf{1}^\top x(t)=1$ for all $t\ge0$;
\item (\emph{Nonnegativity invariance})
      if $x_i(t_0)=0$ for some $t_0$, then $\dot x_i(t_0)\ge0$;
\item (\emph{Strict positivity preservation})
      if $x(0)\in\Delta_n^{\circ}=\{x>0\mid\mathbf{1}^\top x=1\}$,
      then $x(t)\in\Delta_n^{\circ}$ for all finite $t\ge0$.
\end{enumerate}
\end{proposition}

\begin{proof}
(i) Summing~\eqref{eq:generator} over $i$ and exchanging  indices gives
\[
\mathbf{1}^\top \dot x
= \sum_{i,k\neq i}\!\big(Q_{k\to i}(x)x_k - Q_{i\to k}(x)x_i\big)=0,
\]
hence $\mathbf{1}^\top x(t)$ is conserved along trajectories.

(ii) At any boundary point with $x_i=0$,
\[
\dot x_i
=\sum_{k\neq i}\big(Q_{k\to i}(x)x_k - Q_{i\to k}(x)x_i\big)
=\sum_{k\neq i} Q_{k\to i}(x)x_k \;\ge\; 0,
\]
since each term is nonnegative. 
Thus the vector field points inward or tangent to the nonnegative orthant,
and the set $\mathbb{R}^n_{\ge0}$ is forward-invariant.

(iii) Define $s_i(x):=\sum_{k\neq i}Q_{i\to k}(x)\ge0$.
As $Q_{i\to k}$ are continuous on the compact set $\Delta_n$,
there exists $\bar s<\infty$ such that $s_i(x)\le\bar s$ for all $x$.
Removing nonnegative inflow terms yields the differential inequality
\[
\dot x_i
=\sum_{k\neq i} Q_{k\to i}(x)x_k - s_i(x)x_i
\;\ge\; -\,s_i(x)x_i
\;\ge\; -\,\bar s\,x_i.
\]
By Gr\"onwall's inequality \cite{gronwall1919note},
$x_i(t)\ge x_i(0)e^{-\bar s t}>0$ for all finite $t\ge0$,
so trajectories starting in $\Delta_n^{\circ}$ remain strictly positive.
\end{proof}

\begin{remark}
Proposition~\ref{prop:positivity} ensures that the dynamics
are well posed on the open simplex $\Delta_n^{\circ}$, which is
forward-invariant and compact.
This guarantees that the entropy function
$V(x)=\sum_i x_i\log(x_i/v_i)$ is well defined and continuously differentiable
along all trajectories, forming the basis of the subsequent stability
and robustness analysis.
\end{remark}

\begin{remark}
(i) Uniform hypergraphs are recovered by fixing the order~$r$ in~\eqref{eq:Qikx}.
(ii) Directedness and weights correspond to the asymmetry and magnitude of
$a_{i k I}$. (iii) The representation~\eqref{eq:generator} will serve as a
bridge between the tensor formulation and the analytical tools for stability
and robustness developed below.
\end{remark}

\begin{remark}[Comparison with classical graph-based models]
\csx{The present framework differs from classical graph-based conservative dynamics primarily at the modeling level. In traditional graph-based systems, flows are generated directly by pairwise interactions associated with edges, so each transition channel depends only on two-node relations. In contrast, in the proposed hypergraph model each transition rate $Q_{k\to i}(x)$ is generated by tensor-based higher-order interaction terms, meaning that the flow from node $k$ to node $i$ may depend on the joint state of multiple other nodes. Such higher-order interaction mechanisms have also appeared in other hypergraph dynamical systems, including SIS spreading models \cite{cui2025sis} and higher-order Lotka--Volterra systems \cite{cui2025analysis}. As a result, the elementary interaction mechanism is no longer edge-based but group-based, which allows the model to represent collective effects that cannot be reduced to a superposition of pairwise flows. Therefore, the proposed dynamics are formulated in a genuinely higher-order manner, rather than as a direct edge-based extension of a classical graph model.}

\csx{It is also worth noting that certain special classes of higher-order Lotka--Volterra dynamics \cite{cui2025analysis} can be rewritten in the present generator-form framework. In general, higher-order Lotka--Volterra systems describe growth and inhibition effects and are not conservative, so they do not directly fit the simplex-constrained flow-conservation model studied here. However, consider the simplex-constrained replicator-type form with a conservative quantity 
$\dot x_i = x_i\bigl(f_i(x)-\bar f(x)\bigr), \qquad 
\bar f(x)=\sum_{\ell=1}^n x_\ell f_\ell(x), \qquad \sum_{\ell=1}^n x_\ell =1,$
where \(f_i(x)\) is a polynomial function generated by higher-order interactions. If we define, for \(k\neq i\),
$Q_{k\to i}(x):=x_i f_i(x),$
then the generator-form dynamics become
$\dot x_i
=
\sum_{k\neq i}\bigl(Q_{k\to i}(x)x_k-Q_{i\to k}(x)x_i\bigr)
=
\sum_{k\neq i}\bigl(x_i f_i(x)x_k-x_k f_k(x)x_i\bigr).$
Factoring out \(x_i\), we obtain
$\dot x_i
=
x_i\sum_{k\neq i}\bigl(f_i(x)x_k-f_k(x)x_k\bigr)
=
x_i\left(f_i(x)\sum_{k\neq i}x_k-\sum_{k\neq i}x_k f_k(x)\right).$
Using \(\sum_{k\neq i}x_k=1-x_i\) and
$\sum_{k\neq i}x_k f_k(x)=\bar f(x)-x_i f_i(x),$
it follows that
$\dot x_i
=
x_i\Bigl(f_i(x)(1-x_i)-(\bar f(x)-x_i f_i(x))\Bigr)
=
x_i\bigl(f_i(x)-\bar f(x)\bigr).$
Therefore, under suitable positivity conditions on the induced rate kernels, this replicator-type higher-order Lotka--Volterra dynamics can be embedded into the present generator-form model. This shows that the proposed framework includes a meaningful conservative subclass of higher-order Lotka--Volterra-type systems.}

\end{remark}

\subsection{Existence of an entropy-equilibrium structure}

We now establish structural conditions under which
system~\eqref{eq:generator} admits a unique equilibrium in~$\Delta_n$
and possesses an entropy-like Lyapunov function.

\begin{assumption}[{\color{blue}Tensor Generalized Detailed Balance}]\label{ass:gdb}
{\color{blue}
There exists a positive vector $v\in\Delta_n^{\circ}$ such that, for every order $r\in\{1,\dots,R\}$, every pair of distinct $i,k\in[n]$, and every multi-index $I$ associated with the $r$-th tensor (cf.~\eqref{eq:Qikx}),
\begin{equation}\label{eq:gdb-tensor}
v_k\, a^{(r)}_{i k I} = v_i\, a^{(r)}_{k i I}.
\end{equation}
Consequently, the induced rate kernels satisfy, for all $x\in\Delta_n^{\circ}$ and all $i\neq k$,
\begin{equation}\label{eq:gdb}
v_k\, Q_{k\to i}(x) = v_i\, Q_{i\to k}(x).
\end{equation}
}
\end{assumption}

{\color{blue}
Assumption~\ref{ass:gdb} imposes balance not only at the aggregated kernel level, but also at each individual tensor layer (uniform sub-hypergraph). More precisely, as required by~(8), for every order $r$, every pair of distinct nodes $i,k\in[n]$, and every multi-index $I$, the $v$-weighted coefficient of the transition $k\to i$ is exactly balanced by that of the reverse transition $i\to k$, that is \eqref{eq:gdb-tensor}. 
Hence, each uniform hypergraph component (corresponding to one fixed order $r$) is balanced separately, and summing these balanced contributions over all orders yields the aggregated relation \eqref{eq:gdb}
for every state $x$.
}

\begin{lemma}[Existence of equilibrium]\label{lem:existence}
If Assumption~\ref{ass:gdb} holds, then $v\in\Delta_n^{\circ}$
is an equilibrium of system~\eqref{eq:generator}, i.e.,
$f(v;\mathcal{A})=0$.
\end{lemma}

\begin{proof}
Substituting~\eqref{eq:gdb} into~\eqref{eq:generator} gives, for all~$i$,
\[
\sum_{k\neq i}\!\big(Q_{k\to i}(v)v_k - Q_{i\to k}(v)v_i\big)=0,
\]
hence $f(v;\mathcal{A})=0$.
\end{proof}

To establish uniqueness, consider the Kullback-Leibler divergence \cite{kullback1951kullback}
\begin{equation}\label{eq:V}
V(x):=D_{\mathrm{KL}}(x\Vert v)
=\sum_{i=1}^n x_i\log\!\frac{x_i}{v_i}, \qquad x\in\Delta_n,
\end{equation}
and differentiate it along trajectories of~\eqref{eq:generator}.
Using $\dot x_i=\sum_{k\neq i}\big(Q_{k\to i}(x)x_k - Q_{i\to k}(x)x_i\big)$, we obtain
\begin{align}
\dot V(x)
&= \sum_{i=1}^n \frac{\partial V}{\partial x_i}(x)\,\dot x_i
 = \sum_{i=1}^n \Big(\log\frac{x_i}{v_i}+1\Big)\dot x_i \notag\\
&= \sum_{i=1}^n \log\frac{x_i}{v_i}\,\dot x_i
\qquad\text{since } \mathbf{1}^\top\dot x=0. \label{eq:Vdot-step1}
\end{align}
Substituting the dynamics gives
\begin{align}
\dot V(x)
&= \sum_{i=1}^n \log\frac{x_i}{v_i}
   \sum_{k\neq i}\big(Q_{k\to i}(x)x_k - Q_{i\to k}(x)x_i\big) \notag\\
&= \sum_{i\neq k}
   \big(Q_{k\to i}(x)x_k - Q_{i\to k}(x)x_i\big)
   \log\frac{x_i}{v_i}. \label{eq:Vdot-step2}
\end{align}
In the first term of~\eqref{eq:Vdot-step2}, exchange the dummy indices
$(i,k)\mapsto(k,i)$, yielding
\begin{align}
\dot V(x)
&= \sum_{i\neq k}
   Q_{i\to k}(x)x_i
   \Big(\log\frac{x_k}{v_k}-\log\frac{x_i}{v_i}\Big). \notag
\end{align}

{\color{blue}
Let $y_i:=x_i/v_i$. Using $Q_{i\to k}(x)x_i = \big(v_i Q_{i\to k}(x)\big) y_i$ and the symmetry $v_i Q_{i\to k}(x)=v_k Q_{k\to i}(x)$ from Assumption~\ref{ass:gdb}, symmetrizing with respect to $(i,k)$ gives the entropy-dissipation form
\begin{equation}\label{eq:Vdot}
\begin{split}
\dot V(x)
&= -\,\frac{1}{2}\sum_{i\neq k}
   v_i Q_{i\to k}(x)\,
   \Big(\frac{x_i}{v_i}-\frac{x_k}{v_k}\Big)
   \Big(\log\frac{x_i}{v_i}-\log\frac{x_k}{v_k}\Big)\\
&= -\,\frac{1}{2}\sum_{i\neq k}
   v_i Q_{i\to k}(x)\,
   (y_i-y_k)\,(\log y_i-\log y_k)\;\le 0,
\end{split}
\end{equation}
where the inequality follows from $(y_i-y_k)(\log y_i-\log y_k)\ge 0$ for all $y_i,y_k>0$ and $v_i Q_{i\to k}(x)\ge 0$.
}
{\color{blue}
Moreover, $\dot V(x)=0$ holds if and only if $y_i=y_k$ for all pairs $(i,k)$ with $v_i Q_{i\to k}(x)>0$. Under Assumption~\ref{ass:gdb}, this is equivalent to the 
\eqref{eq:gdb} (indeed, $y_i=y_k \Leftrightarrow x_i/x_k=v_i/v_k \Leftrightarrow Q_{k\to i}(x)x_k=Q_{i\to k}(x)x_i$ using $Q_{k\to i}(x)/Q_{i\to k}(x)=v_i/v_k$ implied by \eqref{eq:gdb}).
}

Equality in~\eqref{eq:Vdot} holds if and only if \eqref{eq:gdb} holds
so that the time derivative of~$V$ vanishes precisely on the set
\[
\mathcal{Z}
:=\big\{x\in\Delta_n^\circ \,\big|\, \eqref{eq:gdb} \text{ holds}\big\}.
\]

\begin{assumption}[Uniqueness of detailed-balance solution]\label{ass:db-unique}
The equalities in~\eqref{eq:gdb} admit a unique solution
$x\in\Delta_n^{\circ}$, denoted by $v$; equivalently, $\mathcal{Z}=\{v\}$.
\end{assumption}

\begin{lemma}[Unique equilibrium and entropy structure]\label{lem:unique}
Under Assumptions~\ref{ass:gdb} and~\ref{ass:db-unique},
$v\in\Delta_n^{\circ}$ is the unique equilibrium of~\eqref{eq:generator}.
Moreover, $V(x)$ defined in~\eqref{eq:V} is a strict Lyapunov function.
\end{lemma}

\begin{proof}
By Lemma~\ref{lem:existence}, $v$ is an equilibrium.
Since $\dot V(x)\le 0$ by~\eqref{eq:Vdot} and $\dot V(x)=0$ if and only if
$x\in\mathcal{Z}$, the largest invariant subset of $\{x:\dot V(x)=0\}$
is $\{v\}$ by Assumption~\ref{ass:db-unique}.
LaSalle's invariance principle on the positively invariant simplex~$\Delta_n$
implies that every trajectory converges to $v$, and thus no other equilibrium
exists in~$\Delta_n^{\circ}$.
\end{proof}

\subsection{Lyapunov stability and convergence rate}

The Lyapunov function~\eqref{eq:V} satisfies $\dot V(x)\le 0$ globally and
$\dot V(x)=0$ only at $x=v$. Therefore the pair $(f,V)$ defines an
\emph{entropy-dissipating structure} on the simplex.

\begin{lemma}[Global asymptotic stability]\label{lem:lyapunov}
Under Assumptions~\ref{ass:gdb}--\ref{ass:db-unique},
$V(x)$ is a strict Lyapunov function for~\eqref{eq:generator}, and
the equilibrium $x=v$ is globally asymptotically stable in~$\Delta_n$.
\end{lemma}

\begin{proof}
The result follows directly from~\eqref{eq:Vdot} and
Lemma~\ref{lem:unique} via LaSalle's invariance principle.
\end{proof}

\begin{remark}[Local convergence rate]\label{rem:jac}
The Jacobian $J=\partial f/\partial x|_{x=v}$ satisfies
$\mathbf{1}^\top J=0$ due to mass conservation.
On the tangent subspace $\mathbf{1}^\perp=\{z:\mathbf{1}^\top z=0\}$,
$J$ is Hurwitz under irreducibility of the interaction structure.
The smallest nonzero magnitude of $\mathrm{Re}(\lambda(J))$ defines the
\emph{spectral gap}, which determines the local exponential convergence rate
and plays a key role in the robustness bounds established later.
\end{remark}

\paragraph*{Connectivity of the support graph.}
For a nonnegative matrix $S\in\mathbb{R}_{\ge0}^{n\times n}$ with zero diagonal,
we define its \emph{support graph} as the undirected graph
$\mathcal{G}_S=([n],\mathcal{E}_S)$ whose edge set is
\[
\mathcal{E}_S := \big\{ \{i,k\}\subseteq [n] \;|\; S_{ik}>0\ \text{or}\ S_{ki}>0 \big\}.
\]
We say that $S$ has a \emph{connected support graph} if $\mathcal{G}_S$ is
connected, i.e., every pair of nodes is linked by a path of positive-weight
entries in $S$. \csx{This construction can be viewed as a graph projection of the underlying hypergraph onto an ordinary pairwise graph. Here, however, we only retain the support information and do not assign exact projected edge weights, since for the present purpose we only need to characterize connectivity.}

\begin{proposition}
\label{prop:db-unique-positivity}
Consider the generator-form dynamics \eqref{eq:generator} on the simplex $\Delta_n$ with mass conservation $\mathbf{1}^\top x=1$. 
Assume:

\begin{enumerate}
\item[(a)] \textbf{Positivity:} For every $x\in\Delta_n^\circ$ and every present interaction, $\csx{Q}_{k\to i}(x)>0$.
\item[(b)] \textbf{Irreducibility (strongly connected support graph):} The directed support graph $\mathcal{G}(x)$ induced by $\{\csx{Q}_{k\to i}(x)>0\}$ is strongly connected for all $x\in\Delta_n^\circ$.
\item[(c)] \csx{\textbf{Tensor Generalized Detailed Balance (TGDB):} 
$\pi_k\, a^{(r)}_{i k I} = \pi_i\, a^{(r)}_{k i I}.$
}
\end{enumerate}
Then the detailed-balance equalities
admit a unique solution in $\Delta_n^\circ$, namely
$x \;=\; v \;=\; \frac{\pi}{\mathbf{1}^\top \pi }$.
Consequently, Assumption~\ref{ass:db-unique} holds.
\end{proposition}

\begin{proof}
Fix $x\in\Delta_n^\circ$ satisfying $\csx{Q}_{k\to i}(x)\,x_k=\csx{Q}_{i\to k}(x)\,x_i$ for all $i\neq k$. 
For any edge $(i,k)$ with positive rates in both directions, it follows that
$\frac{x_i}{x_k}
\;=\;
\frac{\csx{Q}_{k\to i}(x)}{\csx{Q}_{i\to k}(x)}{=}
\frac{\pi_i}{\pi_k}.$
Since the directed support graph is strongly connected, by traversing along directed paths we obtain 
$\frac{x_i}{x_k}=\frac{\pi_i}{\pi_k}$ for all node pairs $(i,k)$, hence $x=\alpha\,\pi$ for some $\alpha>0$. 
Using the mass constraint $\mathbf{1}^\top x=1$, we get $\alpha = 1/(\mathbf{1}^\top \pi)$ and therefore $x=v=\pi/(\mathbf{1}^\top \pi)$. 
As $x\in\Delta_n^\circ$, $v\in\Delta_n^\circ$ is indeed admissible. 
Uniqueness follows by the same ratio argument: any other solution must satisfy $x=\alpha\,\pi$ and thus coincide with $v$ after normalization.
\end{proof}

\begin{proposition}[Quadratic entropy dissipation]
\label{prop:quadratic-dissipation}
\csx{Under the assumptions of Proposition~\ref{prop:db-unique-positivity}, suppose in addition that there exists a constant $\underline q>0$ such that
$Q_{k\to i}(x)\ge \underline q$
for all $x\in\Delta_n^\circ$ and all pairs $(i,k)$ belonging to the support graph in Proposition~\ref{prop:db-unique-positivity}. Then there exists a constant $c>0$ such that
$\dot V(x)\le -c\|x-v\|^2,
\qquad \forall x\in\Delta_n^\circ .$
More precisely, one may take
$c=\frac{v_{\min}^2\,\underline q\,\lambda_\star}{2\,v_{\max}^2},$
where
$v_{\min}:=\min_i v_i,\qquad
v_{\max}:=\max_i v_i,$
and
$\lambda_\star:=
\inf_{\substack{z\neq 0\\ v^\top z=0}}
\frac{\sum_{\{i,k\}\in\mathcal E_\star}(z_i-z_k)^2}{\|z\|^2}>0,$
with $\mathcal E_\star$ denoting the edge set of the connected support graph in Proposition~\ref{prop:db-unique-positivity}.}
\end{proposition}

\begin{proof}
\csx{By Assumption~\ref{ass:gdb}, 
\[
\dot V(x)
=
-\frac12\sum_{i\neq k} w_{ik}(x)
\left(\frac{x_i}{v_i}-\frac{x_k}{v_k}\right)
\left(\log\frac{x_i}{v_i}-\log\frac{x_k}{v_k}\right).
\]
Let again
$y_i:=\frac{x_i}{v_i},\qquad z:=y-\mathbf 1 .$
Since $x,v\in\Delta_n^\circ$, one has
$\sum_{i=1}^n v_i y_i=\sum_{i=1}^n x_i=1,$
hence
$v^\top z=0.$
}

\csx{Next, for any $a,b>0$, the inequality
$(a-b)(\log a-\log b)\ge \frac{2(a-b)^2}{a+b}$
holds. Since $x_i\le 1$ and $v_i\ge v_{\min}>0$, we have
$y_i=\frac{x_i}{v_i}\le \frac{1}{v_i}\le \frac{1}{v_{\min}},$
and therefore
$y_i+y_k\le \frac{2}{v_{\min}}.$
It follows that
$(y_i-y_k)(\log y_i-\log y_k)\ge v_{\min}(y_i-y_k)^2.$
Substituting this into the dissipation identity yields
$\dot V(x)
\le
-\frac{v_{\min}}{2}\sum_{i\neq k} w_{ik}(x)(y_i-y_k)^2.$
}
\csx{Now, by the additional uniform positivity assumption and the definition of $w_{ik}(x)$,
$w_{ik}(x)=v_iQ_{i\to k}(x)\ge v_{\min}\underline q$
for all edges $\{i,k\}\in\mathcal E_\star$. Hence
$\dot V(x)
\le
-\frac{v_{\min}^2\underline q}{2}
\sum_{\{i,k\}\in\mathcal E_\star}(y_i-y_k)^2.$}

\csx{By Proposition~\ref{prop:db-unique-positivity}, the graph $\mathcal G_\star=([n],\mathcal E_\star)$ is connected. Therefore, the quadratic form
$\sum_{\{i,k\}\in\mathcal E_\star}(z_i-z_k)^2$
is positive definite on the subspace $\{z\in\mathbb R^n: v^\top z=0\}$, and thus
$\sum_{\{i,k\}\in\mathcal E_\star}(z_i-z_k)^2
\ge
\lambda_\star \|z\|^2,
\quad
\forall z\in\mathbb R^n \text{ with } v^\top z=0.$
Since $y_i-y_k=z_i-z_k$, we obtain
$\dot V(x)\le
-\frac{v_{\min}^2\underline q\,\lambda_\star}{2}\|z\|^2.$}

\csx{Finally, because
$x-v=\operatorname{diag}(v)\,z,$
we have
$\|x-v\|^2
=
\sum_i v_i^2 z_i^2
\le
v_{\max}^2\|z\|^2,$
that is,
$\|z\|^2\ge \frac{1}{v_{\max}^2}\|x-v\|^2.$
Combining the above estimates gives
$\dot V(x)
\le
-\frac{v_{\min}^2\underline q\,\lambda_\star}{2v_{\max}^2}\|x-v\|^2.$
This proves the claim.}
\end{proof}

Now, we want to focus on the following special case of a uniform distribution of the resource.

\begin{proposition}\label{prop:sufficient}
Suppose $Q_{k\to i}(x)=S_{k i}\,\phi(x)$ with $\phi(x)>0$ smooth on
$\Delta_n^{\circ}$, and $S=S^\top\!\ge 0$ has zero diagonal and a connected
support graph. Then Assumptions~\ref{ass:gdb}--\ref{ass:db-unique} hold
with $v=\mathbf{1}/n$.
\end{proposition}

\begin{proof}
Assume $Q_{k\to i}(x)=S_{k i}\,\phi(x)$ with $\phi(x)>0$ on $\Delta_n^{\circ}$ and
$S=S^\top\!\ge 0$ having zero diagonal and a connected support graph.

{\color{blue}
\emph{(i) TGDB for all $x$ with $v=\mathbf{1}/n$.}
Let $v=\mathbf{1}/n$. Then for any $x\in\Delta_n^\circ$ and any $i\neq k$,
\[
v_k\, Q_{k\to i}(x)
= \frac{1}{n}\,S_{k i}\,\phi(x)
= \frac{1}{n}\,S_{i k}\,\phi(x)
= v_i\, Q_{i\to k}(x),
\]
where we used $S_{k i}=S_{i k}$. Hence Assumption~\ref{ass:gdb} holds.
}

\emph{(ii) Uniqueness of the detailed-balance solution.}
The detailed-balance equalities \eqref{eq:gdb} can deduce that
$Q_{k\to i}(x)\,x_k = Q_{i\to k}(x)\,x_i
\;\Longleftrightarrow\; 
S_{k i}\,\phi(x)\,x_k = S_{i k}\,\phi(x)\,x_i
\;\Longleftrightarrow\;
S_{k i}\,x_k = S_{i k}\,x_i.$
By symmetry of $S$, this is equivalent to
$S_{i k}\,(x_i - x_k)=0 \qquad \forall\, i\neq k.$
Hence, for every edge $\{i,k\}$ with $S_{i k}>0$, we must have $x_i=x_k$.
Because the support graph of $S$ is connected, this implies $x_i=\theta$ for all
$i\in[n]$. Using the simplex constraint $\mathbf{1}^\top x=1$ yields
$\theta=1/n$, i.e., $x=\mathbf{1}/n$. Since $\theta>0$, we have
$x\in\Delta_n^{\circ}$. Therefore the solution set of
\eqref{eq:gdb} in $\Delta_n^{\circ}$ is the singleton $\{v\}$, and
Assumption~\ref{ass:db-unique} holds.

Combining (i) and (ii) proves the proposition.
\end{proof}

\begin{remark}[Role of the hypergraph structure]
\csx{The above results highlight that the hypergraph structure plays a genuinely nontrivial role in shaping the equilibrium and stability properties of the system. First, for a non-uniform hypergraph, the interaction structure can be decomposed into several uniform sub-hypergraphs, each corresponding to one tensor layer or one interaction order. Under the TGDB condition, these different layers are not balanced independently around unrelated reference states; rather, they are all required to be compatible with the same vector $v$. In this sense, each uniform sub-hypergraph contributes a local balance relation, and the fact that all layers share the same $v$ expresses a form of cross-order consistency of the higher-order interactions.
Second, the existence and uniqueness of the global equilibrium are determined not by the connectivity of each individual uniform sub-hypergraph, but by the connectivity of the aggregated interaction skeleton, represented here through the overall graph projection (or support graph). This means that the full system may still admit a unique globally asymptotically stable equilibrium even if some individual layers are disconnected, provided that the different layers jointly produce a connected global interaction structure. Therefore, connectivity requirements are effectively relaxed in the hypergraph setting: it is the union of the higher-order interaction channels that matters, rather than the connectivity of each order separately. Taken together, these observations show how local layerwise consistency gives rise to a global dynamical feature. The common equilibrium vector $v$ encodes a shared balance pattern across different interaction orders, while the aggregated graph projection guarantees that this pattern propagates throughout the entire system. Hence, the hypergraph framework does not merely add higher-order terms to an otherwise graph-based dynamics; instead, it allows different orders of interaction to cooperate in producing a coherent global equilibrium and stability structure. From this viewpoint, the results reveal a characteristic feature of hypergraph dynamical systems: global stability may emerge from the cooperation of several individually incomplete layers. A given uniform sub-hypergraph may be insufficient, by itself, to enforce a global behavior, but once combined with other layers sharing the same balanced state, it can become part of an integrated higher-order architecture that does. This illustrates precisely how local structural commonality across layers can generate a global dynamical signature.}
\end{remark}

\begin{remark}[Nonuniqueness example]
\csx{Consider $n=3,r=1$ and define
$Q_{2\to1}(x)=Q_{1\to2}(x)=1,\quad
Q_{k\to i}(x)=0 \ \text{for all other } i\neq k.$
Let
$v=\Big(\tfrac13,\tfrac13,\tfrac13\Big)\in\Delta_3^\circ.$
Then
$v_kQ_{k\to i}(x)=v_iQ_{i\to k}(x),\qquad \forall\, i\neq k,$
so Assumption~\ref{ass:gdb} holds. However, the generator dynamics reduce to
$\dot x_1=x_2-x_1,\qquad
\dot x_2=x_1-x_2,\qquad
\dot x_3=0.$
Hence every point of the form
$x^\ast=(a,a,1-2a),\qquad 0<a<\tfrac12,$
is an equilibrium in $\Delta_3^\circ$. Therefore the detailed-balance equations admit infinitely many solutions, and Assumption~\ref{ass:db-unique} fails. This shows that the GDB condition alone does not guarantee uniqueness of the equilibrium; additional connectivity-type assumptions, such as those in Proposition~\ref{prop:db-unique-positivity}, are essential.}
\end{remark}

\section{Robustness to Topological Perturbations: Sensitivity and Local ISS Bounds}

\subsection{Perturbed dynamics and definitions}

Having established the entropy-dissipating Lyapunov-like function of the unperturbed system,
we now examine how this structure behaves under small perturbations of the
coupling tensor and external inputs.

\paragraph{Perturbed tensor system}
Consider the tensor-coupled dynamics
\begin{equation}\label{eq:perturbed-system}
    \dot x = f(x;\mathcal{A}+\Delta\mathcal{A}) + w(t),
\end{equation}
where $\Delta\mathcal{A}$ represents a structural (topological) perturbation of
the coupling tensor $\mathcal{A}$, and $w(t)\in\mathbb{R}^n$ denotes an additive
input or modeling error. The perturbation $\Delta\mathcal{A}$ may capture small
changes in hyperedge weights, directions, or membership relations, while
$w(t)$ accounts for external noise, approximation errors, or neglected dynamics.

We assume throughout that $\mathbf{1}^\top w(t)=0$ so that the total mass is
preserved, ensuring $\mathbf{1}^\top \dot x = 0$ and $x(t)\in\Delta_n$ for all~$t$.

\paragraph{Perturbed equilibrium}
Let $v$ be the unique equilibrium of the nominal system $f(x;\mathcal{A})=0$
under the generalized detailed balance (GDB) condition.
The perturbed system~\eqref{eq:perturbed-system} admits an equilibrium
$\tilde v$ satisfying
\begin{equation}\label{eq:perturbed-equilibrium}
    f(\tilde v;\mathcal{A}+\Delta\mathcal{A}) = 0.
\end{equation}
When $\Delta\mathcal{A}=0$ and $w(t)\equiv 0$, one recovers $\tilde v=v$.

\paragraph{Problem formulation}
The goal of this section is twofold:
\begin{enumerate}
    \item[(i)] to quantify how the equilibrium point shifts under small
    topological perturbations $\Delta\mathcal{A}$, i.e.,
    to establish bounds on $\|\tilde v-v\|$ (static robustness);
    \item[(ii)] to analyze how the trajectories of
    system~\eqref{eq:perturbed-system} deviate from the nominal equilibrium~$v$
    in the presence of both $\Delta\mathcal{A}$ and $w(t)$, using the
    Kullback--Leibler Lyapunov function~$V(x)=D_{\mathrm{KL}}(x\Vert v)$
    (dynamic robustness).
\end{enumerate}

\paragraph{Notation and assumptions}
Throughout this section, we use the following conventions:
\begin{itemize}
    \item $J := \frac{\partial f}{\partial x}(v;\mathcal{A})$
    denotes the Jacobian of the nominal system at~$v$.
    Under the GDB condition, $J$ is Hurwitz on $\mathbf{1}^\perp$
    with spectral gap $c_{\textbf{gap}}>0$ as defined in
    Remark~\ref{rem:jacobian-gap}.
    \item $\partial_{\mathcal{A}}f(v;\mathcal{A})[\Delta\mathcal{A}]$
    denotes the Fr\'echet derivative of~$f$ with respect to~$\mathcal{A}$,
    evaluated at~$(v,\mathcal{A})$ and applied to~$\Delta\mathcal{A}$.
    \item The perturbation amplitude $\|\Delta\mathcal{A}\|$ is assumed
    sufficiently small so that all linearizations and higher-order terms
    $\mathcal{O}(\|\Delta\mathcal{A}\|^2)$ are well defined.
\end{itemize}
The constant $c>0$ appearing below in the quadratic entropy-dissipation and local ISS
bounds denotes the nonlinear decay rate furnished by
Proposition~\ref{prop:quadratic-dissipation}. The spectral gap
$c_{\textbf{gap}}$ is the corresponding local linear contraction rate of
$J|_{\mathbf{1}^{\perp}}$. They are directly linked, and coincide in the
linearized symmetric setting.

The subsequent subsections derive (i) an explicit first-order sensitivity bound
for $\|\tilde v-v\|$ and (ii) local input-to-state stability inequalities that
quantify the dissipativity of the perturbed system under~$V(x)$.

\subsection{Equilibrium sensitivity analysis}

We first analyze how the equilibrium point changes when the coupling tensor
$\mathcal{A}$ is slightly perturbed. The following result quantifies the
first-order variation of the equilibrium under the generalized detailed
balance (GDB) structure.

\begin{proposition}[First-order equilibrium sensitivity]\label{prop:sensitivity}
Let $v\in\Delta_n^\circ$ be the nominal equilibrium of the generator-form
dynamics $\,\dot x=f(x;\mathcal{A})\,$ under Assumption~\ref{ass:gdb}.
Write $J:=\frac{\partial f}{\partial x}(v;\mathcal{A})$ and let
$J|_{\mathbf{1}^{\perp}}$ denote its restriction to the tangent subspace
$\mathbf{1}^{\perp} = \{ z \in \mathbb{R}^n \mid \mathbf{1}^{\top} z = 0 \}$,
where $J|_{\mathbf{1}^{\perp}}$ is invertible (Hurwitz) by Lyapunov argument.
Consider a small perturbation $\Delta\mathcal{A}$ of the coupling tensor and
let $\tilde v\in\Delta_n^\circ$ be the perturbed equilibrium satisfying
$f(\tilde v;\mathcal{A}+\Delta\mathcal{A})=0$.
Then, with the orthogonal projector $P:=I-\frac{1}{n}\mathbf{1}\mathbf{1}^\top$,
\begin{equation}\label{eq:sens-firstorder}
\tilde v - v
= -\,\big(J|_{\mathbf{1}^{\perp}}\big)^{-1}\,
   P\,\partial_{\mathcal{A}} f(v;\mathcal{A})[\Delta\mathcal{A}]
\;+\; \mathcal{O}(\|\Delta\mathcal{A}\|^2),
\end{equation}
where $\partial_{\mathcal{A}} f(v;\mathcal{A})[\cdot]$ is the Fr\'echet
derivative of $f$ w.r.t.\ $\mathcal{A}$ at $(v,\mathcal{A})$.

Consequently, there exists $L_{\mathcal{A}}>0$ such that
\begin{equation}\label{eq:sens-norm}
\|\tilde v - v\|
\;\le\;
\big\|\big(J|_{\mathbf{1}^{\perp}}\big)^{-1}\big\|\;
L_{\mathcal{A}}\;\|\Delta\mathcal{A}\|
\;+\; \mathcal{O}(\|\Delta\mathcal{A}\|^2),
\end{equation}
i.e., the equilibrium displacement scales linearly with the perturbation
magnitude, with gain set by the resolvent norm of the Jacobian on
$\mathbf{1}^{\perp}$.
\end{proposition}

\begin{proof}
Let $v\in\Delta_n^\circ$ be the nominal equilibrium satisfying $f(v;\mathcal{A})=0$
(Assumption~\ref{ass:gdb}). Because of mass conservation, any equilibrium lies
on the simplex $\Delta_n$ and, in particular, for the perturbed equilibrium
$\tilde v$ we have $\mathbf{1}^\top \tilde v=\mathbf{1}^\top v=1$, hence
$\tilde v - v \in \mathbf{1}^{\perp}$.

\smallskip
\noindent\textit{Step 1 (reduction to the tangent space).}
Let $P:=I-\frac{1}{n}\mathbf{1}\mathbf{1}^\top$ be the orthogonal projector onto
$\mathbf{1}^{\perp}$, and select an orthonormal basis matrix
$U\in\mathbb{R}^{n\times (n-1)}$ with columns spanning $\mathbf{1}^{\perp}$
so that $U^\top U=I$ and $UU^\top=P$.
Define the reduced vector field
\[
F(y,\mathcal{A})\;:=\;U^\top f\!\big(v+Uy\;;\;\mathcal{A}\big),
\qquad y\in\mathbb{R}^{n-1}.
\]
Then $F(0,\mathcal{A})=U^\top f(v;\mathcal{A})=0$ and
\[
\frac{\partial F}{\partial y}(0,\mathcal{A})
= U^\top \frac{\partial f}{\partial x}(v;\mathcal{A})\,U
=: J_T,
\]
where $J=\frac{\partial f}{\partial x}(v;\mathcal{A})$.
By irreducibility and the entropy structure, $J$ is Hurwitz on
$\mathbf{1}^{\perp}$, hence $J_T$ is nonsingular.

\smallskip
\noindent\textit{Step 2 (implicit function theorem on the manifold).}
Consider a small perturbation $\Delta\mathcal{A}$ and let
$\tilde y\in\mathbb{R}^{n-1}$ denote the solution of
\[
F(\tilde y,\mathcal{A}+\Delta\mathcal{A}) = 0
\quad\Longleftrightarrow\quad
U^\top f\!\big(v+U\tilde y\;;\;\mathcal{A}+\Delta\mathcal{A}\big)=0.
\]
By the implicit function theorem at $(y,\mathcal{A})=(0,\mathcal{A})$
and the nonsingularity of $J_T$, there exists a unique smooth map
$\tilde y(\cdot)$ with $\tilde y(\mathcal{A})=0$ and, for sufficiently small
$\|\Delta\mathcal{A}\|$,
\[
\tilde y(\mathcal{A}+\Delta\mathcal{A})
= -\,J_T^{-1}\,U^\top\,
\partial_{\mathcal{A}} f(v;\mathcal{A})[\Delta\mathcal{A}]
\;+\; \mathcal{O}(\|\Delta\mathcal{A}\|^2).
\]

\smallskip
\noindent\textit{Step 3 (lifting back to the simplex and first-order term).}
Since $\tilde v = v + U\tilde y$, we obtain
\[
\tilde v - v
= -\,U J_T^{-1} U^\top\,
\partial_{\mathcal{A}} f(v;\mathcal{A})[\Delta\mathcal{A}]
\;+\; \mathcal{O}(\|\Delta\mathcal{A}\|^2).
\]
Using $U J_T^{-1} U^\top = \big(J|_{\mathbf{1}^{\perp}}\big)^{-1} P$
(as the inverse on $\mathbf{1}^{\perp}$ and zero on $\mathrm{span}\{\mathbf{1}\}$),
we rewrite the expansion as
\begin{equation*}
\tilde v - v
= -\,\big(J|_{\mathbf{1}^{\perp}}\big)^{-1}
   P\,\partial_{\mathcal{A}} f(v;\mathcal{A})[\Delta\mathcal{A}]
\;+\; \mathcal{O}(\|\Delta\mathcal{A}\|^2),
\end{equation*}
which is the claimed first-order formula on the tangent space of the simplex.

\smallskip
\noindent\textit{Step 4 (norm bound).}
Taking norms and using submultiplicativity yields
\[
\|\tilde v - v\|
\;\le\;
\big\|\big(J|_{\mathbf{1}^{\perp}}\big)^{-1}\big\|\;
\big\|\partial_{\mathcal{A}} f(v;\mathcal{A})\big\|\;
\|\Delta\mathcal{A}\|
\;+\; \mathcal{O}(\|\Delta\mathcal{A}\|^2),
\]
which proves the sensitivity bound with
$L_{\mathcal{A}}:=\|\partial_{\mathcal{A}} f(v;\mathcal{A})\|$.
\end{proof}



\begin{remark}[Singularity and spectral gap of the Jacobian]\label{rem:jacobian-gap}
The Jacobian $J=\partial f/\partial x|_{x=v}$ always satisfies
$J\mathbf{1}=0$ because of the mass-conservation constraint
$\mathbf{1}^\top x=1$.
Hence $J$ is singular on $\mathbb{R}^n$ but invertible on the tangent
subspace $\mathbf{1}^{\perp}=\{z\in\mathbb{R}^n\mid \mathbf{1}^\top z=0\}$
of the simplex~$\Delta_n$.
All sensitivity and stability results
(Proposition~\ref{prop:sensitivity} and Theorem~\ref{thm:iss})
are therefore interpreted on this subspace, where
$J|_{\mathbf{1}^{\perp}}$ is Hurwitz under irreducibility.

The \emph{spectral gap} is defined as
\begin{equation}\label{eq:spectral-gap}
    \csx{c_{\textbf{gap}} }:= \min_{\lambda\in\sigma(J|_{\mathbf{1}^{\perp}})}
    \big(-\mathrm{Re}(\lambda)\big) > 0,
\end{equation}
which quantifies the slowest exponential decay rate of perturbations
orthogonal to~$\mathbf{1}$.  
If $J$ is symmetric (e.g., when $Q_{k\to i}(x)=Q_{i\to k}(x)$ at equilibrium),
then all eigenvalues are real and $c_{\textbf{gap}}$ coincides with the smallest nonzero one.

In the tensor-coupled hypergraph setting, the linearization at the equilibrium $v$
contracts each higher-order interaction tensor with $v$ and with the perturbation
direction, producing \emph{effective pairwise weights}. Concretely, entries of
the Jacobian $J=\partial f/\partial x|_{x=v}$ collect partial derivatives of the
higher-order flows; this operation projects multi-node (hyperedge) influences onto
pairwise couplings that govern the \emph{local} response. After projecting $J$ onto
$\mathbf{1}^{\perp}$ (to remove the mass-conserving uniform mode), these aggregated
pairwise terms form an \emph{effective Laplacian matrix} on the tangent space.
In the special case of purely pairwise interactions, this reduces
to the standard graph Laplacian and
$c_{\textbf{gap}}$ becomes its second smallest eigenvalue in the classical analysis. 
Hence $c_{\textbf{gap}}$ generalizes the algebraic connectivity to higher-order hypergraphs,
indicating how strongly the system resists and dissipates perturbations.
\end{remark}

\paragraph{Interpretation}
Equation~\eqref{eq:sens-norm} shows that the equilibrium deviation
$\|\tilde v - v\|$ scales linearly with the perturbation amplitude
$\|\Delta\mathcal{A}\|$, with proportionality factor
$\big\|\big(J|_{\mathbf{1}^{\perp}}\big)^{-1}\big\|L_{\mathcal{A}}$.
The inverse norm $\big\|\big(J|_{\mathbf{1}^{\perp}}\big)^{-1}\big\|^{-1}$
represents the smallest contraction rate of the linearized dynamics on
the tangent subspace $\mathbf{1}^{\perp}$, which is bounded below by the
\emph{spectral gap} \eqref{eq:spectral-gap}.
A larger spectral gap~$c_{\textbf{gap}}$ implies faster local convergence and
smaller sensitivity to structural perturbations, providing a quantitative
link between dynamical stability and robustness.

\paragraph{Practical relevance}
In tensor-coupled or hypergraph systems, $\Delta\mathcal{A}$ may represent
local rewiring, edge-weight noise, or slight symmetry breaking.
Proposition~\ref{prop:sensitivity} provides a first-order robustness metric:
networks with larger convergence gaps (smaller $\|J^{-1}\|$) are less
susceptible to equilibrium drift when their topology is perturbed.

The next subsection extends this static analysis to dynamic robustness,
establishing local ISS inequalities for the trajectories
of the perturbed system. The locality comes from estimating the perturbation
terms on a compact neighborhood of~$v$; by contrast, the nominal entropy
dissipation in Proposition~\ref{prop:quadratic-dissipation} remains global
under its stronger uniform-positivity hypothesis.

\subsection{Local input-to-state stability under entropy Lyapunov}

We now extend the previous static sensitivity result to a dynamic
robustness analysis.  The aim is to characterize how the trajectories of the
perturbed system~\eqref{eq:perturbed-system} deviate from the nominal
equilibrium~$v$ under small topological perturbations $\Delta\mathcal{A}$ and
bounded external inputs~$w(t)$.

\begin{theorem}[\csx{Local ISS inequality under entropy Lyapunov}]\label{thm:iss}
Let Assumption~\ref{ass:gdb} hold, and let $\Omega\Subset\Delta_n^\circ$ be a
compact neighborhood of~$v$. Suppose that in the unperturbed case there exists
$c>0$ such that
\[
\dot V(x)\le -c\|x-v\|^2,
\qquad \forall x\in\Omega.
\]
This assumption is in particular satisfied if
Proposition~\ref{prop:quadratic-dissipation} holds on all of
$\Delta_n^\circ$.
Consider the perturbed system~\eqref{eq:perturbed-system} with
$\mathbf{1}^\top w(t)=0$ and sufficiently small $\Delta\mathcal{A}$.
Then there exist constants $C_1,C_2,\eta_\Omega>0$ such that, for every
trajectory that remains in~$\Omega$,
\begin{equation}\label{eq:Vdot-perturbed}
    \dot V(x)
    \;\le\;
    -\,\frac{c}{2}\,\|x-v\|^2
    + C_1\,\|\Delta\mathcal{A}\|^2
    + C_2\,\|w(t)\|^2 .
\end{equation}
Consequently, for all $t\ge 0$,
\begin{equation}\label{eq:ISS-bound}
    V(x(t))
    \le e^{-\eta_\Omega t} V(x(0))
    + \frac{C_1}{\eta_\Omega}\|\Delta\mathcal{A}\|^2
    + \frac{C_2}{\eta_\Omega}\!\!\sup_{s\le t}\|w(s)\|^2 .
\end{equation}
In particular, because there exist $m_\Omega,M_\Omega>0$ such that
\[
m_\Omega\|x-v\|^2 \le V(x)\le M_\Omega\|x-v\|^2,
\qquad \forall x\in\Omega,
\]
the equilibrium $v$ satisfies \csx{a local ISS estimate with respect to the
disturbances $\Delta\mathcal{A}$ and $w$}.
\end{theorem}

\begin{proof}
Define the Lyapunov function as in \eqref{eq:V}.
Since $\Omega$ is a compact subset of $\Delta_n^\circ$ and
$\nabla^2 V(x)=\mathrm{diag}(1/x_1,\dots,1/x_n)$, there exists
$\kappa_\Omega>0$ such that
\[
\|\nabla V(x)-\nabla V(v)\|
\le \kappa_\Omega \|x-v\|,
\qquad \forall x\in\Omega.
\]
Because $\nabla V(v)=0$, this yields
\[
\|\nabla V(x)\|\le \kappa_\Omega \|x-v\|,
\qquad \forall x\in\Omega.
\]
Moreover, since $V$ is $C^2$ in a neighborhood of~$v$ with positive definite
Hessian $\nabla^2V(v)=\mathrm{diag}(1/v_1,\dots,1/v_n)$, the ratio
$x\mapsto V(x)/\|x-v\|^2$ extends continuously to~$x=v$. Because $\Omega$ is
compact, there exist $m_\Omega,M_\Omega>0$ such that
\[
 m_\Omega \|x-v\|^2 \le V(x)\le
 M_\Omega\|x-v\|^2,
\qquad \forall x\in\Omega.
\]

We have 
\begin{align}
\dot V(x)
=& \nabla V(x)^{\!\top} f(x;\mathcal{A})
 + \nabla V(x)^{\!\top}
  \nonumber\\& \big(f(x;\mathcal{A}+\Delta\mathcal{A})
      - f(x;\mathcal{A})\big)
 + \nabla V(x)^{\!\top} w(t).
\end{align}
The first term is bounded by $-c\|x-v\|^2$ from the unperturbed case.
Since $f(\cdot;\mathcal{A})$ is locally Lipschitz in~$\mathcal{A}$ and~$x$,
there exists $L_{\mathcal{A},\Omega}>0$ such that
\[
\|\!f(x;\mathcal{A}+\Delta\mathcal{A})-f(x;\mathcal{A})\!\|
 \le L_{\mathcal{A},\Omega}\,\|\Delta\mathcal{A}\|,
 \qquad \forall x\in\Omega.
\]
Using Cauchy--Schwarz \cite{aldaz2015advances}, we obtain
\[
|\nabla V(x)^{\!\top}
   (f(x;\mathcal{A}+\Delta\mathcal{A})-f(x;\mathcal{A}))|
 \le \kappa_\Omega\,L_{\mathcal{A},\Omega}\,\|x-v\|\,\|\Delta\mathcal{A}\|.
\]
Applying Young's inequality \cite{nielsen1994sharpness},
\[
ab\le \frac{\varepsilon}{2}a^2+\frac{1}{2\varepsilon}b^2,
\qquad \varepsilon>0,
\]
with $a=\|x-v\|$, $b=\kappa_\Omega L_{\mathcal{A},\Omega}\|\Delta\mathcal{A}\|$,
and $\varepsilon=c/2$, gives
\[
\kappa_\Omega\,L_{\mathcal{A},\Omega}\|x-v\|\|\Delta\mathcal{A}\|
 \le \frac{c}{4}\|x-v\|^2
    + \frac{\kappa_\Omega^2 L_{\mathcal{A},\Omega}^2}{c}\|\Delta\mathcal{A}\|^2.
\]
Similarly,
$\nabla V(x)^{\!\top}w(t)\le \frac{c}{4}\|x-v\|^2
  + \frac{\kappa_\Omega^2}{c}\|w(t)\|^2$.
Combining these bounds yields~\eqref{eq:Vdot-perturbed} with
$C_1=\frac{\kappa_\Omega^2 L_{\mathcal{A},\Omega}^2}{c}$ and
$C_2=\frac{\kappa_\Omega^2}{c}$.
Using $V(x)\le M_\Omega\|x-v\|^2$ on~$\Omega$, we further obtain
\[
\dot V(x)\le -\frac{c}{2M_\Omega}V(x)
  + C_1\|\Delta\mathcal{A}\|^2
  + C_2\|w(t)\|^2.
\]
Setting $\eta_\Omega:=\frac{c}{2M_\Omega}$ and applying Gr\"onwall's
inequality \cite{gronwall1919note} yields~\eqref{eq:ISS-bound}. Since $V$ and
$\|x-v\|^2$ are locally equivalent on~$\Omega$, this proves the claimed local
ISS estimate around~$v$.
\end{proof}

\paragraph{Interpretation}
Inequality~\eqref{eq:Vdot-perturbed} reveals that the entropy function~$V(x)$
remains a Lyapunov function for the perturbed system, up to additive terms that
scale quadratically with the perturbation magnitudes.  The exponential factor
in~\eqref{eq:ISS-bound} describes the nominal local convergence toward~$v$, while the
two constant terms represent the residual error floors induced by
$\Delta\mathcal{A}$ and $w(t)$.  When both vanish, the estimate reduces to the
nominal decay law, and the perturbed dynamics collapse to the unperturbed
system.

\paragraph{Physical implication}
The spectral gap~$c_{\textbf{gap}}$ captures the intrinsic rate of local entropy dissipation.
Networks with larger~$c_{\textbf{gap}}$ exhibit faster convergence and smaller sensitivity to
topological or external disturbances.  In practice, the steady-state bias
$\|\tilde v-v\|$ predicted by Proposition~\ref{prop:sensitivity} coincides with
the residual floor implied by~\eqref{eq:ISS-bound}, providing a consistent
measure of static and dynamic robustness.

\subsection{Discussion and implications}

The results derived above demonstrate that the entropy-dissipating structure of
tensor-coupled dynamics is preserved under small perturbations of the coupling
tensor and under bounded external inputs.  The equilibrium shift and the local ISS
inequality jointly provide a quantitative description of the system's
robustness properties.

\paragraph{Spectral interpretation}
Both the sensitivity bound~\eqref{eq:sens-norm} and the local ISS inequality
\eqref{eq:ISS-bound} depend on the spectral properties of the Jacobian
$J=\partial f/\partial x|_{x=v}$.  The smallest nonzero magnitude of
$\mathrm{Re}(\lambda(J))$ determines the \emph{spectral gap}~$c_{\textbf{gap}}$, which governs
the exponential convergence rate of the unperturbed system and the attenuation
of perturbations.  Larger spectral gaps correspond to stronger contraction in
the subspace $\mathbf{1}^{\perp}$ and thus to higher structural robustness.

\paragraph{Entropy dissipation under perturbations}
Even in the presence of $\Delta\mathcal{A}$ and $w(t)$, the relative entropy
$V(x)=D_{\mathrm{KL}}(x\Vert v)$ remains nonincreasing except for a small
offset proportional to $\|\Delta\mathcal{A}\|^2$ and $\|w\|^2$.  Hence, the
perturbed system retains a \emph{dissipative envelope}: trajectories converge
toward a neighborhood of the nominal equilibrium whose radius is of order
$\mathcal{O}(\|\Delta\mathcal{A}\|+\|w\|)$.

\paragraph{Static versus dynamic robustness}
The equilibrium sensitivity (Proposition~\ref{prop:sensitivity}) quantifies how
the steady state itself is displaced by topological changes, while the local ISS
inequality (Theorem~\ref{thm:iss}) characterizes the transient response to
time-varying disturbances.  Together, they form a unified view of robustness:
the same spectral and geometric features of the generator that ensure rapid
convergence also determine how the system tolerates perturbations.

\paragraph{Design implications}
These analytical insights suggest that enhancing the spectral gap of the
linearized generator---for example, by strengthening inter-hyperedge
connectivity or balancing flow asymmetries---can improve both convergence speed
and resilience to uncertainty.  The developed framework therefore offers a
principled tool for the analysis and design of robust higher-order networks.

\paragraph{Numerical validation}
In the following section, we illustrate the theoretical findings using
synthetic and real hypergraph examples.  We examine how the equilibrium
distance $\|\tilde v-v\|$ scales with the perturbation norm
$\|\Delta\mathcal{A}\|$ and verify that the time evolution of the entropy
function~$V(x(t))$ obeys the predicted local ISS envelope.

\section{Numerical experiments}

To illustrate the theoretical results, we perform the following experiments.

\subsection{Numerical Simulations}

We consider the generator-form dynamics \eqref{eq:generator} on the simplex
$\Delta_n$ with $n=8$ agents. The base matrix $S\in\mathbb{R}^{n\times n}_{\ge0}$ is symmetric, nonnegative, and has a zero diagonal, encoding the pairwise interaction intensities among nodes. The state-dependent transition rates are given by
\begin{equation}
    Q_{k\to i}(x) = S_{ki}\Big(1 + \alpha \sum_{j=1}^{n} x_j^2\Big),
    \qquad \alpha = 0.8,
\end{equation}
which preserve the generalized detailed balance at the uniform equilibrium $v = \mathbf{1}/n$.  
The vector field is integrated using a fourth-order Runge--Kutta method with projection onto the simplex $\Delta_n$ at each substep to avoid numerical drift. The initial condition is a random point in $\Delta_n$, and the time step is $\Delta t = 10^{-2}$.

\paragraph*{a) Baseline case}
Figure~\ref{fig:state_baseline} (\textit{Numerical case: state trajectory}) shows the time evolution of all state components without perturbations.  
The trajectories monotonically converge to the uniform equilibrium $v$, confirming the global asymptotic stability of the detailed-balanced dynamics.

\paragraph*{b) Perturbed case}
Topological perturbations are modeled by replacing $S$ with $S + \Delta S$, where $\Delta S$ is a zero-diagonal, randomly generated noise matrix scaled to a prescribed Frobenius norm $\|\Delta S\|_F \simeq 0.3$ and then clipped to preserve nonnegativity.  
In addition, zero-sum sinusoidal inputs $w(t)$ with amplitude $\rho = 0.03$ are applied.  
Figure~\ref{fig:state_perturbed} (\textit{Numerical case: state trajectory under perturbation}) shows that all trajectories remain bounded and close to the equilibrium manifold, exhibiting small steady-state oscillations.  
This demonstrates the robustness of the system against both structural and external perturbations.

\paragraph*{c) Equilibrium sensitivity}
To quantify the equilibrium sensitivity predicted by Proposition~\ref{prop:sensitivity}, the perturbation magnitude is swept over nine levels in $[0, 0.6]$.  
For each level, the dynamics are integrated to steady state, and the resulting equilibrium shift $\|\tilde v - v\|_2$ is recorded.  
Figure~\ref{fig:equil_shift} (\textit{Equilibrium sensitivity}) presents the results in log--log scale, showing an approximately linear relation between $\|\tilde v - v\|_2$ and $\|\Delta S\|_F$, which agrees with the first-order sensitivity bound.

\begin{figure}[t]
    \centering
    \includegraphics[width=0.7\linewidth]{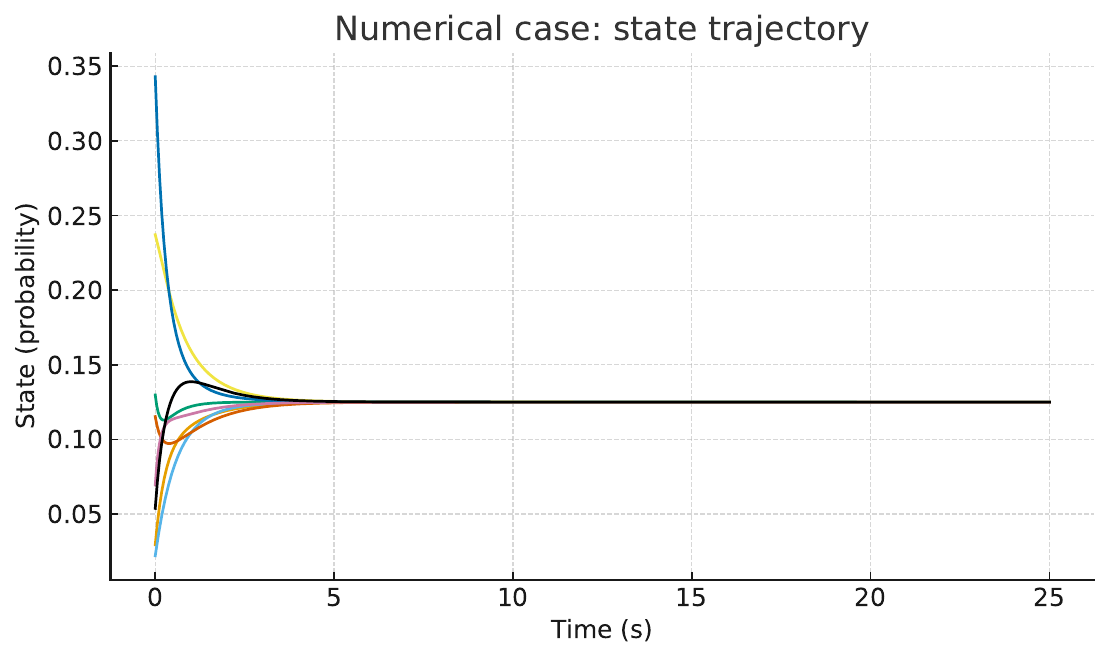}
    \caption{Numerical case: state trajectory.}
    \label{fig:state_baseline}
\end{figure}

\begin{figure}[t]
    \centering
    \includegraphics[width=0.7\linewidth]{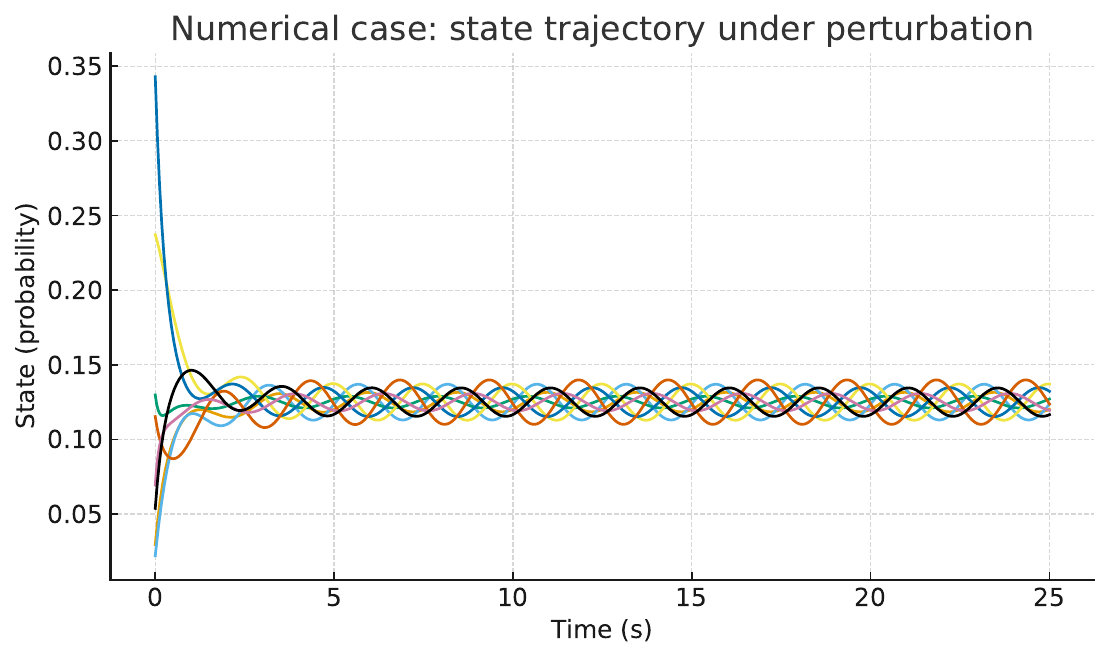}
    \caption{Numerical case: state trajectory under perturbation.}
    \label{fig:state_perturbed}
\end{figure}

\begin{figure}[t]
    \centering
    \includegraphics[width=0.7\linewidth]{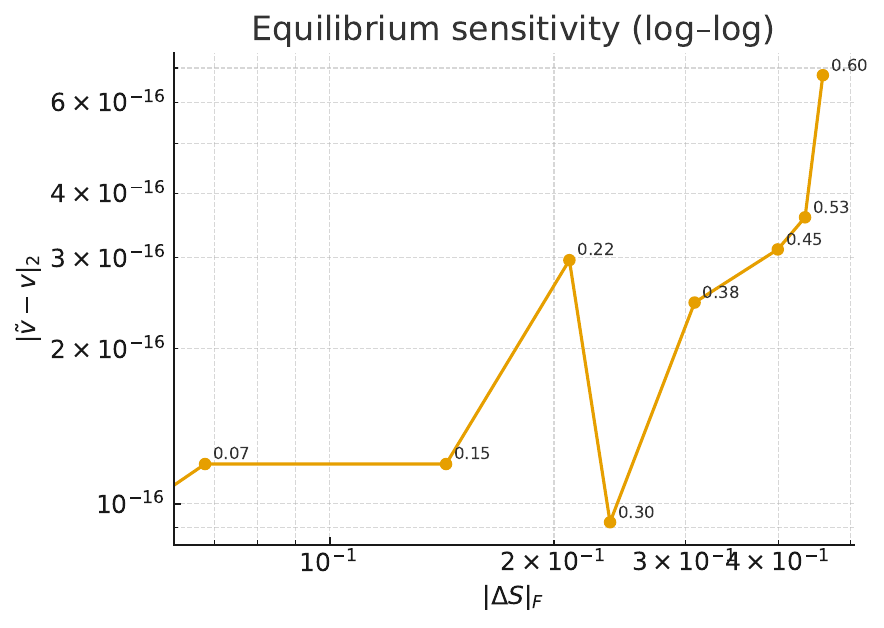}
    \caption{Equilibrium sensitivity (log--log).}
    \label{fig:equil_shift}
\end{figure}

\subsection{Multi-Agent System Simulation}

To further demonstrate the applicability of the proposed generator-form conservative framework to cooperative networked systems, we consider a planar multi-agent system modeled by double-integrator dynamics. For $i\in\{1,\ldots,m\}$, the dynamics are
\begin{align}
   & \dot p_i = v_i, \nonumber \\
   & \dot v_i = u_i, \nonumber 
\end{align}
where $p_i,v_i\in\mathbb{R}^2$ denote the position and velocity of the $i$-th agent, respectively.
Each agent follows a momentum-conserving, generator-type coupling rule given by
\begin{equation} \label{multi}
    u_i = -k_p \!\sum_{k=1}^{m} L_{ik}(x)\,p_k
           -k_d \!\sum_{k=1}^{m} L_{ik}(x)\,v_k + w_i(t),
\end{equation}
where $L(x)=D(x)-W(x)$ is the state-dependent Laplacian matrix with
\begin{equation}
    W(x) = S\Big(1+\alpha(\|p\|_F^2+\|v\|_F^2)\Big),
\end{equation}
and $S\in\mathbb{R}_{\ge0}^{m\times m}$ is a symmetric, zero-diagonal base matrix describing the interaction topology. $p=(p_1^T,p_2^T,\cdots, p_m^T)^T$, $v=(v_1^T,v_2^T,\cdots,v_m^T)^T$. 
It can be interpreted as: the coupling weight $W(x)$ varies with the overall motion energy of the overall system.
This design guarantees $\sum_i u_i = 0$, which ensures conservation of the total momentum $\sum_i v_i$ of the entire multi-agent system, consistent with the generator-form conservation principle.

\begin{figure}[t]
    \centering
    \includegraphics[width=0.7\linewidth]{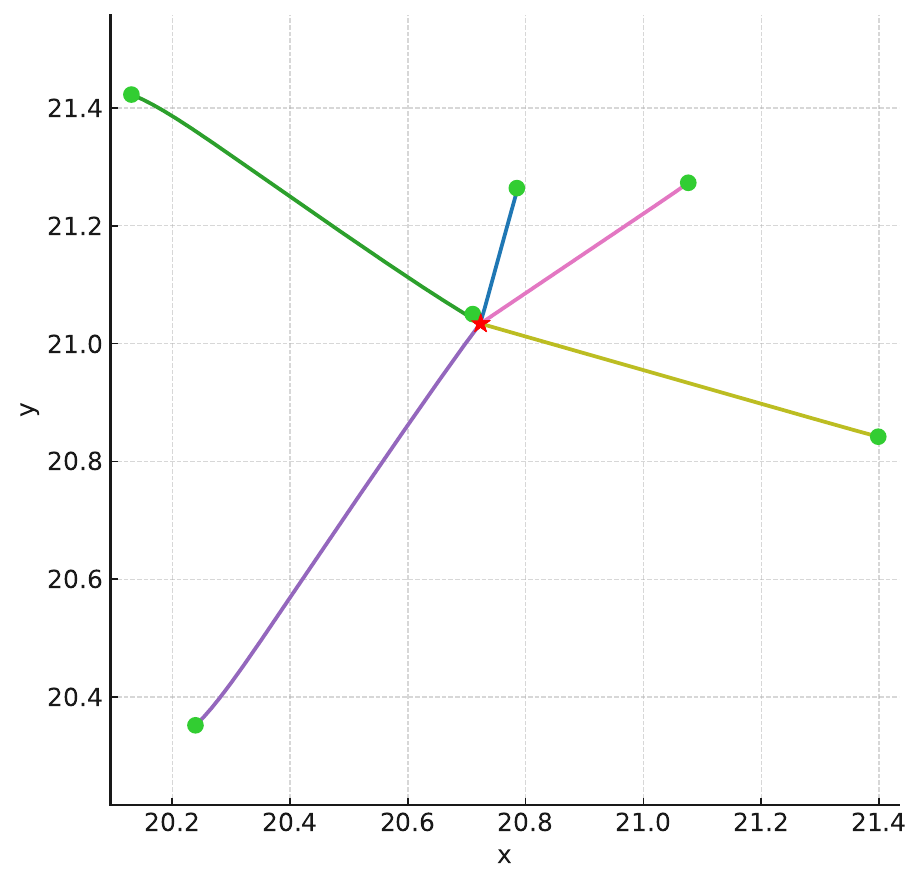}
    \caption{Trajectories of multi-agent system on conservative hypergraphs.}
    \label{fig:agent_traj}
\end{figure}

\begin{figure}[t]
    \centering
    \includegraphics[width=0.7\linewidth]{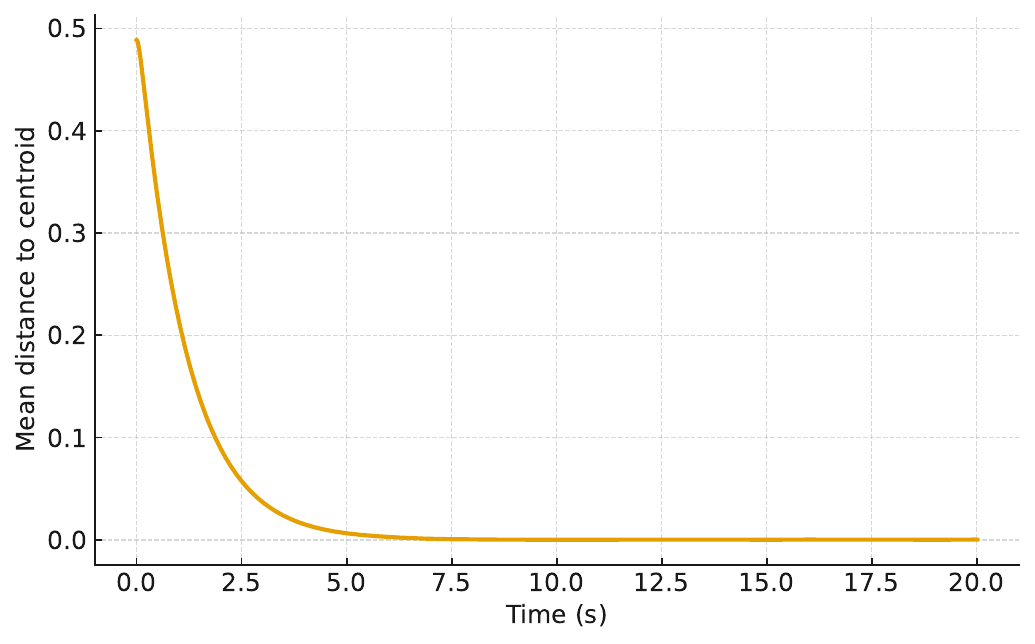}
    \caption{Mean distance of multi-agent system on conservative hypergraphs.}
    \label{fig:agent_dist}
\end{figure}


The simulation parameters are chosen as $m=6$, $k_p=1.0$, $k_d=1.2$, $\alpha=0.02$, and the integration step $\Delta t=0.02\,\mathrm{s}$.  
The base topology $S$ is randomly generated to be connected and symmetric.  
The agents are initially scattered around $[-8,4]$ with small random offsets and zero initial velocities.  
A zero-sum external disturbance is applied as $w_1(t)=-w_2(t)=[0.6\sin(2\pi 0.25t),0]^{\!\top}$ and $w_i(t)=0$ for $i>2$, ensuring that the net external input remains zero.

Figure~\ref{fig:agent_traj} illustrates the planar trajectories of all agents.  
The group collectively converges toward the centroid while maintaining total momentum conservation.
Figure~\ref{fig:agent_dist} depicts the mean distance to the centroid, which decays rapidly before reaching a small oscillatory steady state under the zero-sum input.


\section{Conclusion}
\csx{This paper develops an entropy-based framework for the stability and robustness analysis of tensor-coupled flow-conservation dynamics on hypergraphs. By formulating the system in generator form on the simplex, we establish global asymptotic stability under a generalized detailed-balance condition, with the Kullback--Leibler divergence as a strict Lyapunov function. We further relate entropy dissipation, convergence, and robustness through the spectral gap of the Jacobian, and derive both first-order equilibrium sensitivity bounds and a local ISS estimate under perturbations. Numerical examples confirm the theory and illustrate the applicability of the framework to conservatively coupled multi-agent systems. }

\bibliographystyle{IEEEtran}
\bibliography{bib}

\vspace{-1 cm}
\begin{IEEEbiography}[{\includegraphics[width=1.4in,height=1.35in,clip,keepaspectratio]{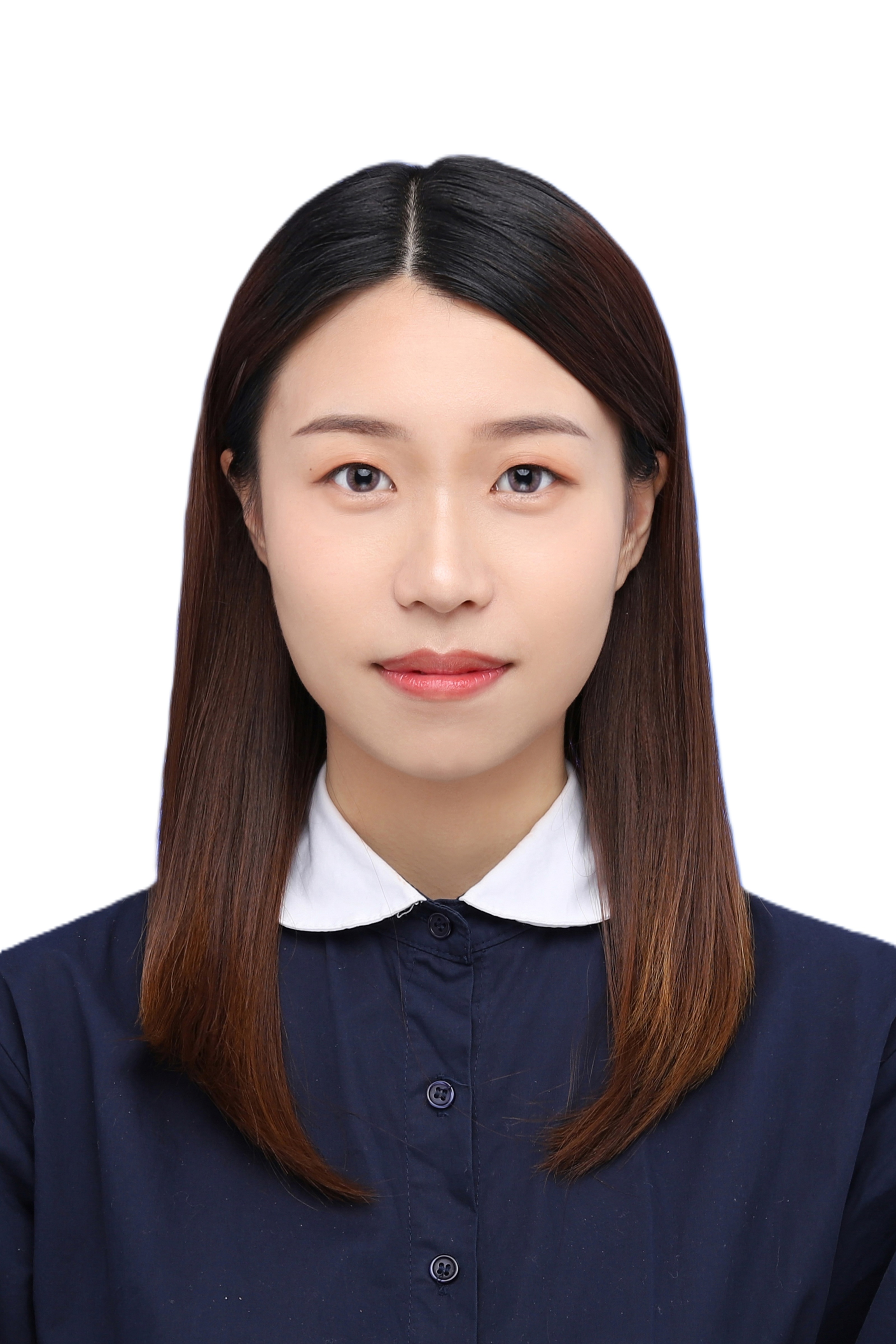}}]{Chencheng Zhang}
received the B.S. degree in automation from the College of Information Engineering, Yangzhou University, Yangzhou, China, in 2017. She was a visiting
 Ph.D. student at the University of Groningen from 2022 to 2023.
She received the Ph.D. degree in control theory and control engineering from Nanjing University of Aeronautics and Astronautics, Nanjing, China, in
2024, where she is currently a Post-Doctoral Researcher.
Her current research focuses on analysis of higher-order network systems, fault estimation and fault-tolerant control for networked control systems and helicopter applications.
\end{IEEEbiography}

\vspace{-1 cm}
\begin{IEEEbiography}[{\includegraphics[width=1in,height=1.25in,clip,keepaspectratio]{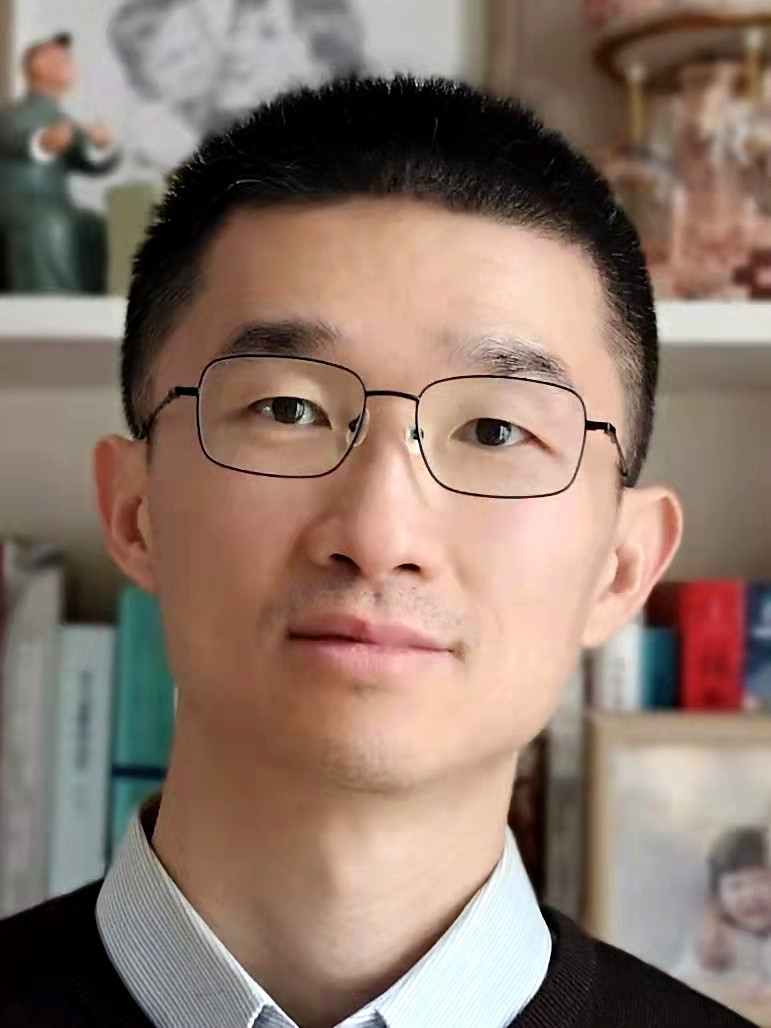}}]{Hao Yang}
(M'13-SM'15) received the B.Sc. degree in electrical automation from Nanjing Tech University, Nanjing, China, in 2004, and the Ph.D. degrees in automatic
control from Universit\'e de Lille: Sciences et Technologies, Lille, France, and Nanjing University of Aeronautics and Astronautics (NUAA), Nanjing, China, both in 2009. Since 2010, he has been working with College of Automation Engineering in NUAA, where he has been a Full Professor since 2015. His research interest includes control, optimization, game and fault tolerance of switched and network systems with their aerospace applications.

Dr. Yang was the recipient of the National Science Fund of China for Excellent Young Scholars in 2016, and the Top-Notch Young Talents of Central Organization Department of China in 2017. He has served as Associate Editor for Nonlinear Analysis: Hybrid Systems, Cyber-Physical Systems, Acta Automatica Sinica, and Chinese Journal of Aeronautics. He is a member of the IFAC Technical Committee on Fault Detection, Supervision and Safety of Technical Processes.
\end{IEEEbiography}

\vspace{-1 cm}
\begin{IEEEbiography}[{\includegraphics[width=1in,height=1.25in,clip,keepaspectratio]{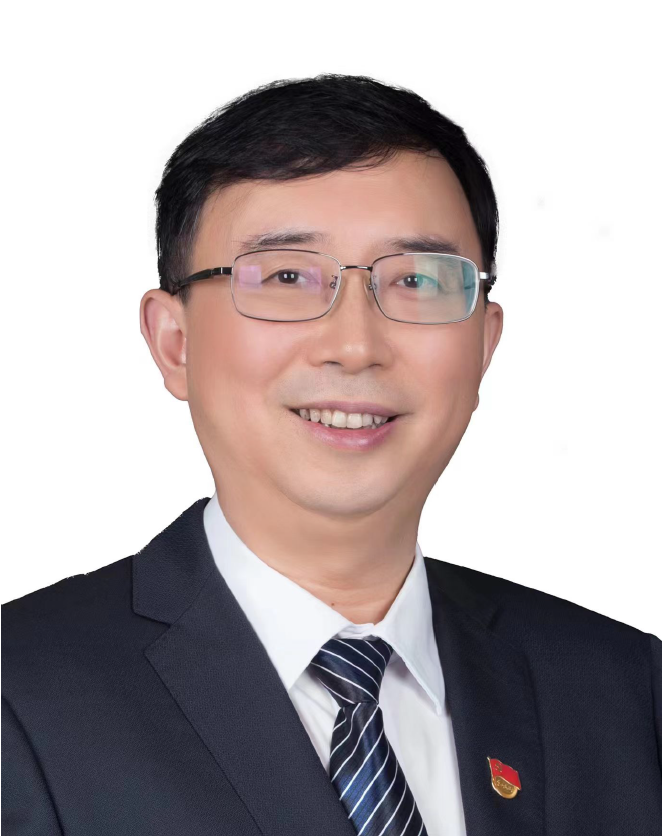}}]{Bin Jiang}
(M'03-SM'05-F'20) received the Ph.D.degree in automatic control from Northeastern University, Shenyang, China, in 1995. He had ever been a Post-Doctoral Fellow, a Research Fellow,
an Invited Professor, and a Visiting Professor in Singapore, France, USA, and Canada, respectively.
He is currently the President of Nanjing University of Aeronautics and Astronautics, Nanjing, China and the Chair Professor of Cheung Kong Scholar Program with the Ministry of Education.
His current research interests include intelligent fault diagnosis and fault tolerant control and their applications to helicopters, satellites, and high-speed trains. He has authored eight books in the related
fields.

Dr. Jiang was a recipient of the Second Class Prize of National Natural Science Award of China. He is a Fellow of Chinese Association of Automation (CAA). He currently serves as an Editor of International Journal of Control, Automation and Systems, and an Associate Editor or an Editorial Board Member for a number of journals, such as IEEE Trans. on Cybernetics, Journal of the Franklin Institute, Neurocomputing, etc. He is a Chair of Control Systems Chapter in IEEE Nanjing Section, and a member of IFAC Technical Committee on Fault Detection, Supervision, and Safety of Technical Processes.
\end{IEEEbiography}

\vspace{-1 cm}
\begin{IEEEbiography}[{\includegraphics[width=1in,height=1.25in,clip,keepaspectratio]{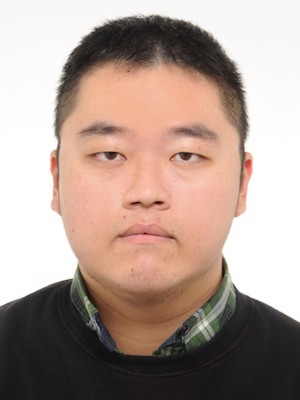}}]{Shaoxuan Cui} is currently a postdoc at Engineering and Technology Institute Groningen (ENTEG), the University of Groningen, the Netherlands. He received his Ph.D. degree from the Bernoulli Institute for Mathematics, Computer Science and Artificial Intelligence of the University of Groningen, the Netherlands. He received his B.Sc. degree in Electrical Engineering and Information Technology from the Technical University of Kaiserslautern, Germany, and Fuzhou University, China in 2018, and his M.Sc. Degree in Electrical Engineering and Information Technology from the Technical University of Munich, Germany, in 2020. His research interests include modeling, analysis, and control of networked systems.
\end{IEEEbiography}

\end{document}